\documentclass[11pt]{scrartcl}

\makeatletter

\def\blfootnote{\gdef\@thefnmark{}\@footnotetext}

\makeatother

\usepackage[layout=a4paper,lmargin=2cm,rmargin=2cm]{geometry}

\usepackage{type1cm}       

\usepackage{makeidx}         
\usepackage{graphicx}       
                             
\usepackage{multicol}        
\usepackage[bottom]{footmisc}

\usepackage{newtxtext}      
\usepackage{newtxmath}       
\usepackage{changes}

\usepackage{xcolor}
\definecolor{Red}{HTML}{E53E30} 
\definecolor{Green}{HTML}{00AD69} 
\definecolor{Blue}{HTML}{2171b5}
\definecolor{Purple}{HTML}{652F6C} 

\usepackage{tikz}
\usepackage{fontawesome}

\usetikzlibrary{fadings}

\newcommand{\R}{{\mathbb R}}
\newcommand{\C}{{\mathbb C}}

\newcommand{\diag}{\text{diag}}
\newcommand{\Id}{\text{Id}}

\newcommand{\innerproduct}[1]{ \langle #1 \rangle }

\newcommand{\vertiii}[1]{{\left\vert\kern-0.25ex\left\vert\kern-0.25ex\left\vert #1 
		\right\vert\kern-0.25ex\right\vert\kern-0.25ex\right\vert}}

\DeclareMathOperator{\trace}{tr}
\DeclareMathOperator{\rank}{rk}

\usepackage{amsmath}
\usepackage[thmmarks,thref]{ntheorem}

\newtheorem{theorem}{Theorem}[section]
\newtheorem{lemma}{Lemma}[section]
\newtheorem{corollary}{Corollary}[section]
\newtheorem{definition}{Definition}[section]
\newtheorem{proposition}{Proposition}[section]

\theoremstyle{plain}
\theoremheaderfont{\bfseries\itshape}
\theorembodyfont{\upshape}
\theoremsymbol{\ensuremath{\square}}

\newtheorem*{proof}{Proof} 
\newtheorem*{sproof}{Proof sketch}

\usepackage{appendix}
\usepackage{caption}

\makeindex             

\usepackage{bibentry}
\nobibliography* 
\usepackage[resetlabels]{multibib}
\newcites{LG}{References}

\begin{document}
\title{Proof methods for robust low-rank matrix recovery}
% Use \titlerunning{Short Title} for an abbreviated version of
% your contribution title if the original one is too long
\author{Tim Fuchs\thanks{\quad Corresponding author: {tim.fuchs@tum.de}}\; \thanks{\quad Department of Mathematics, Technical University of Munich, Germany}, \ David Gross\thanks{\quad Institute for Theoretical Physics, University of Cologne, Germany}, \ Peter Jung\thanks{\quad Communications and Information Theory Group, Technische Universität Berlin, Germany}\; \thanks{\quad Data Science in Earth Observation, Technical University of Munich, Germany}, \ Felix Krahmer\footnotemark[2],\\ \ Richard Kueng\thanks{\quad Institute for Integrated Circuits, Johannes Kepler University Linz, Austria}, \ Dominik Stöger\thanks{\quad Ming Hsieh Department of Electrical and Computer Engineering, University of Southern California, USA}}

\maketitle

\abstract{
Low-rank matrix recovery problems arise naturally as mathematical formulations of various
inverse problems, such as matrix completion, blind deconvolution, and phase retrieval. 
 Over the last two decades, a number of works have rigorously analyzed the reconstruction performance for such scenarios, giving rise to a rather general understanding of the potential and the limitations of low-rank matrix models in sensing problems. 
  In this article, we compare the two main proof techniques that have been paving the way to a rigorous analysis, discuss their potential and limitations, and survey their successful applications.
  On the one hand, we review approaches  based on descent cone analysis, showing that they often lead to strong  guarantees even in the presence of adversarial noise, but face limitations when it comes to structured observations.
  On the other hand, we discuss techniques using approximate dual certificates and the golfing scheme, which are often better suited to deal with practical measurement structures, but sometimes lead to weaker guarantees.
  Lastly, we review recent progress towards analyzing descent cones also for structured scenarios---exploiting the idea of splitting the cones into multiple parts that are analyzed via different techniques.
}

\clearpage
\tableofcontents
\clearpage
\section{Introduction}

Computationally tractable data acquisition in high dimensions is a fundamental problem in various real-world applications in signal processing, data science, and physics.
Nyquist sampling or scanning the data in full is often unfeasible.
This motivates the use of compressive observation schemes, which employ regularization methods to revover as much of the signal as possible from seemingly incomplete observations.
Thus, quantifying the trade-off between sample complexity and reconstruction accuracy has become a key task for identification of feasible regimes and the design of efficient approaches for sensing and reconstruction.
These questions have been central in the area of inverse problems for many years. 
However, starting in the early 2000's, a highly successful novel viewpoint has emerged. 
Namely initiated by the influential works on compressed sensing \cite{Candes2005,candes:stablesignalrecovery}, various authors have studied the problem of what can be gained when the measurements 
can be optimized over all vectors or within a structural measurement framework \cite{foucart2013mathematical}.
Commonly, the term compressed sensing is used nowadays also for more general sensing scenarios beyond the initial setup that follow this paradigm. 

In this generality, compressed sensing is therefore concerned with the recovery of structured data, i.e., data that lives on a low-dimensional subset embedded in a high-dimensional space, from a number of 
observations that scales with the intrinsic dimension, rather than the
ambient dimension. 
As it turns out, for a large class of different
measurement models combined with various structural constraints,
choosing the free parameters of the measurement scheme at random leads to near-optimal performance. 

Initially, the model of a non-trivial but relevant low-dimensional set was given by \emph{sparse vectors}.
For matrices, a natural basis-independent notion of sparsity is ``sparsity in the eigenbasis'', i.e.\ \emph{low rank}, and we are thus led to studying the \emph{low-rank matrix recovery problem}:
Estimate an unknown $n_1\times n_2$-matrix $X_0$ from $m$ observations of the form
\begin{align}\label{eqn:measurement model}
y= \mathcal{A} \left(X_0\right) + e \in \mathbb{C}^m,
\end{align}
where $\mathcal{A}$ is a known linear measurement operator and $e$ is additive noise. 
Here and in the following, we will use the notation
\begin{align}\label{eq:measurement-operator}
\mathcal{A}(X_0)(i) := \langle A_i, X_0 \rangle \quad A_i \in \mathbb{C}^{n_1 \times n_2}
\end{align}
which expresses the $i$-th component of the measurement as the Frobenius inner product with a matrix $A_i$.
The problem is interesting in the regime 
\begin{align*}
  \rank(X_0) \max \{n_1, n_2\} \leq m \ll n_1 n_2,
\end{align*}
where $\mathrm{rk}(X_0)$ denotes the matrix rank.
Assume for concreteness that we have the bound
$\Vert e \Vert_2 \le \tau $ on the noise strength ($\|\cdot\|_2$
denotes the $\ell_2$-norm of a vector).
As $X_0$ has low-rank one could naively try to estimate $X_0$ by solving the following minimization problem
\begin{align*}
\begin{split}
\underset{X \in \mathbb{C}^{n_1 \times n_2}}{\text{minimize}} \quad  & \rank \left(X\right) \\
\text{subject to} \quad & \Vert  \mathcal{A} \left(X\right) -y  \Vert_2 \le \tau. 
\end{split}
\end{align*}
Unfortunately, problems of this type are NP-hard in general, as
minimizing the support size of a vector (i.e.\ finding the sparsest solution) can be
considered as a special case \cite{natarajan1995sparse}. Therefore, in \cite{fazel2001rank} it was proposed to use the nuclear
norm $\Vert \cdot \Vert_{\ast} $ (sum of singular values) as a proxy for the rank.  
For this reason, 
the following approach was has been suggested
\cite{candes2009exact,candes2010power,gross2011recovering,recht2011simpler,chen_coherenceoptimal}
for matrix completion:
\begin{align}\label{SDP_MC}
\begin{split}
\underset{X \in \mathbb{C}^{n_1 \times n_2}}{\text{minimize}} \quad  &\Vert X \Vert_{\ast} \\
\text{subject to} \quad & \Vert  \mathcal{A} \left(X\right) -y  \Vert_2 \le \tau. 
\end{split}
\end{align}
It is the analysis of this semi-definite program (SDP) we are concerned with in the present article.
Before tackling the technical details, we briefly list three important applications of the framework of low-rank matrix recovery.

\subsection{Sample applications}
In this section we highlight three famous applications of low-rank matrix recovery which have been investigated intensively in the last years.

\subsubsection{Matrix completion}\label{subsub:matrixcompletion}
Maybe the most natural instantiation of the general model (\ref{eq:measurement-operator}) is the case where the measurements reveal individual matrix elements
\begin{align}\label{equ:MCopdefinition}
  \mathcal{A} \left(X\right) \left(i\right)
  := 
  \sqrt{ \tfrac{n_1 n_2}{m}} 
  \innerproduct{X,e_{a_i}e^*_{b_i}}_F
  = 
  \sqrt{ \tfrac{n_1 n_2}{m}} 
  X_{a_i, b_i},
\end{align}
where $\{e_{a_i}\}$ and $\{e_{b_i} \}$ denote the standard basis of $\mathbb{C}^{n_1}$ and $\mathbb{C}^{n_2}$, respectively.
This is the \emph{matrix completion problem}. 
Since it arises in many different applications such as multiclass
learning \cite{argyriou2008convex}, collaborative filtering
\cite{collabfiltering} and as distance matrix completion problem in
sensor localization tasks \cite{Javanmard2013}, see here also Figure
\ref{fig:localization}, it has become very popular in the last decade and has been studied intensively in the statistics, machine learning, and signal processing literature.

\begin{figure}
  \begin{center}
   \includegraphics[width=0.85\linewidth]{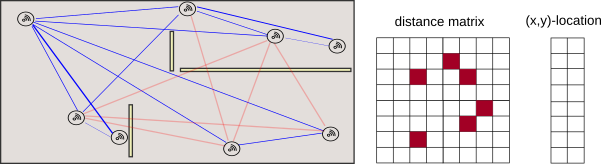}
 \end{center}
  \captionsetup{width=0.85\linewidth}
 \caption{{\em Distance matrix completion:}   
One source of low-rank matrices are \emph{Gram matrices}
   which encode the Euclidean geometries of point sets.  The task of
   recovering Gram matrices is related to the one of finding
   \emph{distance matrices} from few pairwise distances. This problem
   appears for example in sensor localization, see
   e.g. \cite{Javanmard2013}. Wireless sensors are distributed in an
   area (indoor room, industry hall etc.) and measure individual
   signal strengths but obstacles like walls block certain path
   directions (red).
   The goal is to complete the matrix of pairwise distances and
   compute the sensor locations.
   Recall that given $n$ points
   $\{x_i\}_{i=1}^n\subset\mathbb{R}^d$, 
   the Gram matrix $G_{i,j} = \langle x_i, x_j\rangle$ has rank
   upper-bounded by $d$, independent of $n$.}
\label{fig:localization}
\end{figure}

Assume that the matrix elements $(a_i, b_i)$ for $i\in[m]:=\{1\dots m\}$ to be revealed are chosen independently and uniformly among all $n_1 \times n_2$ possibilities.
It is clear that not \emph{all} low-rank matrices can be efficiently recovered from few such measurements.
For example, if $X$ has a single non-zero entry, then unless $m=O(n_1 n_2)$, the probability that any non-zero information is obtained is small.

To identify a set of well-behaved instances,
Ref.~\cite{candes2009exact} introduced the following two \emph{coherence parameters} 
\begin{align*}
%\mu^2 \left( X_0 \right) := \max \left\{ \mu_1^2 \left(U\right), \mu_1^2 \left(V\right)  \right\}.
\mu \left( U \right)&:= \sqrt{\tfrac{n_1}{r}} \underset{i \in \left[n_1\right]}{\max} \Vert U^* e_i \Vert_2 \\
\mu \left( V \right)&:= \sqrt{\tfrac{n_2}{r}} \underset{i \in \left[n_2\right]}{\max} \Vert V^* e_i \Vert_2
\end{align*}
where $X_0=U\Sigma V^*$ with $U\in \C^{n_1 \times r} $ and $V\in \C^{n_2 \times r} $   denotes the singular value decomposition (SVD). 
Indeed, it was shown in \cite{candes2010power} that 
\begin{align}
   m\gtrsim n_1 r \log n_1 \max \left\{ \mu^2 \left( U \right), \mu^2 \left( V \right)    \right\} \label{eq:necessary}
\end{align} 
observations are necessary for an rank-$r$ matrix $X_0$ to be uniquely determined from the revealed entries.\\
Subsequently, a series of works \cite{gross2011recovering,recht2011simpler} established that this sampling rate is almost sufficient as well.
Compared to Eq.~\eqref{eq:necessary}, an additional $\log (n)$-factor and a third incoherence parameter suffice to ensure exact recovery via nuclear norm minimization~\eqref{SDP_MC}. See Section~\ref{sec:golfing} for a detailed statement and proof sketch.

\subsubsection{Blind deconvolution}\label{section:blinddeconv}

Blind deconvolution \cite{Haykin1994} refers to the problem of
recovering a signal $x \in \mathbb{C}^L$ from the noisy convolution $w\ast x + e \in \mathbb{C}^L$,
where $w \in \mathbb{C}^L$ is an unknown kernel and $e \in \mathbb{C}^L$ refers to additive noise. When
using appropriate cyclic extensions or considering zero-padding the
convolution can rewritten as a circular convolution%, which for $w \in \mathbb{C}^L $ and $ x \in \mathbb{C}^L $ is defined by
\begin{align}
(w\ast x)(i):= \sum_{j=1}^L w_j x_{i-j} \quad \text{for $i\in[L]$.} \label{eq:convolution}
\end{align}
The difference $i-j $ is considered modulo $L$. As prototypical
example of a bilinear inverse problem, blind deconvolution refers to
recovering $(x,w)$ from a noisy version of $w\ast x$ and the precise role of $x$ and $w$ depends
on the underlying application.
In imaging, for example, one considers such a problem for the
two-dimensional convolution. The signal $x$ typically
represents the image and $w$ is an unknown blurring kernel
\cite{Stockham1975a}. In communication engineering, the discrete model
above describes the effective convolution in complex baseband. Hence, $w$
represents the sampled impulse response of the transmission channel
and the task is to demodulate and decode information from the signal
vector $x$, only having access to the noisy channel output $w\ast x+e$, see
Figure \ref{fig:bd}. The conventional coherent approach in this application is
to send known pilot signals, to first estimate $w$ and then demodulate
later information-bearing signals $x$. However, this approach is not
feasible for short signals $x$ and for communication at low latency or
high mobility. For communication engineers, the important
question is then how much overhead is required for coping with the
unknown impulse response $w$ of the communication channel
\cite{Godard1980} when using non-coherent strategies \cite{Walk:TWC:MOCZ:basicprinciples}.

Of course, blind deconvolution is a highly underdetermined bilinear inverse
problem.  Without further
assumptions, recovery is only possible up to inherent ambiguities
\cite{Choudhary2013,walk2016ambiguities}. To avoid non-trivial ambiguities one
has to further constrain the vectors, for example by assuming that $x$
and $w$ lie in $N$ and $K$-dimensional subspaces,
respectively. As we will outline below, this yields to the problem of
recovering $N\times K$ matrices from $L$ observations. To be more
compliant with existing works in the literature we will stick to this
notation implying that $n_1=N$, $n_2=K$ and $m=L$. In formulas, we assume that $w=B h_0$ and $x= C \overline{m}_0$ for given $ B \in
\C^{L \times K}$, $ C \in
\C^{L \times N}$ and unknown $h_0\in \C^{K}$, $m_0\in \C^{N}$. Then the measurement operator acting on $h_0$ and $m_0$ is known to be generically injective up to the unavoidable scaling ambiguity if and only if $L\geq 2(N+K-4)$ \cite{LLB15, KK17}. That is, one aims for sampling complexities that are near-linear in $N+K$.

Following \cite{ahmed2014blind}, we
consider the case that $B$ is a fixed matrix such that $B^* B = \Id$ and $C$ is a random matrix with i.i.d complex normal entries $ C_{ij} \overset{\textit{iid}}{\sim} \mathcal{CN} (0, 1/\sqrt{L}) $.
The choice of the \textit{random} matrix $C$ is motivated by the
success of randomization in compressed sensing as well as by
applications in wireless communications. Here $m$ contains a message
to be transmitted and $C$ is a coding matrix. The signal
$x=C\overline{m}_0$ gets transmitted through a time-invariant channel,
which can be modeled as a circular convolution with impulse response
$w$ when using an appropriate cyclic prefix.

In many applications it is reasonable to assume that only the first
few entries of $w$ are non-zero as the path delays are often much
shorter than the length of the signals $x$. In this case $B$ would be
the matrix which extends $h_0 \in \mathbb{C}^K$ by zeros. Hence, the
receiver observes $w\ast x +e$, where $e$ represents additive noise,
and the goal is to reconstruct the original message contained in
the vector $m_0$.\\
\begin{figure}[ht]
	\centering
	\includegraphics[width=.7\linewidth]{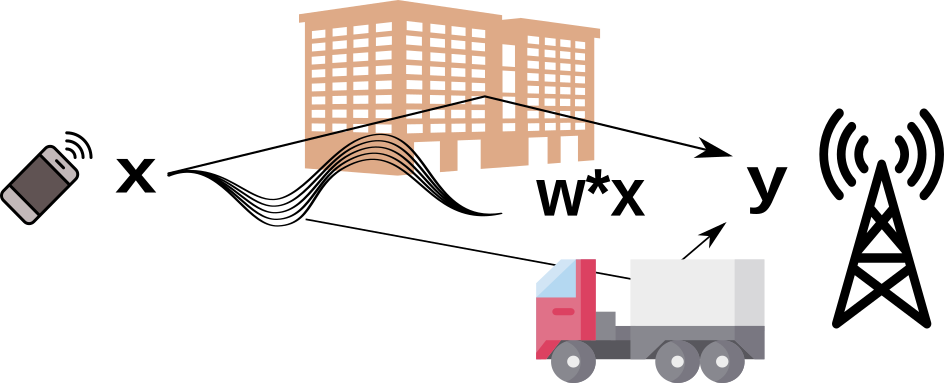}
	\captionsetup{width=0.85\linewidth}
	\caption{\emph{Blind deconvolution}: a wireless device transmits a
          signal $x$ received by the basestation. Due to reflections in the
          environment the signal experiences an unknown channel
          distortion, represented by a convolution $w\ast x$ with the
          impulse response $w$.}
	\label{fig:bd}
\end{figure}

Now let $F\in \C^{L \times L}$ be the unitary discrete Fourier
transformation matrix. It is well-known that $F$ diagonalizes the
circular convolution, i.e.,
\begin{align*}
\widehat{w \ast x} := F \left( w\ast x \right)  = \sqrt{L}  \diag\left(FBh_0\right)FC\overline{m}_0.
\end{align*}
Let $b_{\ell}$ denote the $\ell$th row of $\overline{   FB } $ and let $c_{\ell}$ denote the $\ell$th row of $\sqrt{L}FC$. Note that this implies that all the entries of $ \left\{ c_{\ell} \right\}_{\ell=1}^L $ are jointly independent and have distribution $ \mathcal{CN} \left(0,1\right) $.  Moreover, we obtain that
\begin{align*}
\left( \widehat{w \ast x} \right)_{\ell}= b^*_{\ell} h_0 m^*_0 c_{\ell}  = \innerproduct{ b_{\ell} c^*_{\ell}, h_0 m^*_0 }.
\end{align*}
We observe that $\widehat{w \ast x}  $ is linear in the $K\times N$
matrix $h_0 m^*_0$. This motivates the definition of the linear operator $\mathcal{A}: \C^{K \times N} \rightarrow \mathbb{C}^L $ by
\begin{align}\label{def:Ablinddeconv}
\left(\mathcal{A} \left(X\right) \right) \left(\ell\right) := \innerproduct{b_{\ell} c_{\ell}^*, X} \quad \text{where 
$\ell \in[L]$.}
\end{align}
Hence, we obtain the model
\begin{align*}
y:=\widehat{w \ast x} +e= \mathcal{A} \left(X_0\right) + e,
\end{align*}
where $X_0=h_0 m^*_0$ and $e\in \C^L$ represents noise with $\Vert e \Vert_2 \le \tau $. Note that $X_0$ is a rank-one matrix. 
This reformulation effectively reduces blind deconvolution to a low rank matrix recovery problem, where measurement matrices correspond to outer products $A_{\ell} = b_\ell c_\ell^*$.

If in addition, a sparsity constraint is to be imposed, the problem becomes considerably more difficult. In particular, linear combinations of the convex regularizers no longer lead to sample-efficient recovery guarantees \cite{oymak2015simultaneously} even when using optimal tuning \cite{Kliesch2019}. Only under additional structural assumptions, recovery guarantees are available using an alternating minimization approach
\cite{lee2017blind,geppert2018sparse}. This, however, is beyond the scope of this article.

\subsubsection{Phase retrieval}

\begin{figure}
    \centering
\begin{tikzpicture}[baseline,scale=0.7]
%\draw[fill=gray] (-4.25,-1.5) rectangle (-4,1.5);
\draw[] (6,-3.375) rectangle (6.5,3.375);
\foreach \y in {-3.125,-2.875,...,3.125}
{\draw (6,\y) -- (6.5,\y);};
\draw[<-] (6.75,0) -- (7.25,0);
\node at (7.5,0) {$y_k$};% = \left| \langle f_k D^*, x \rangle \right|^2+e_k$};
\node[text width=3cm] at (6,-4.25) {detector in image plane};
%%%%
%\draw[Blue,thick] (-4,0.5) -- (6,3.375);
%\draw[Blue,thick] (-4,-0.5) -- (6,-3.375);
%%%
\fill[Blue,opacity=0.75,path fading=north] (-8,0) rectangle (-4,0.5);
\fill[Blue,opacity=0.75,path fading=south] (-8,0) rectangle (-4,-0.5);
\fill[Blue,opacity=0.5,path fading=north] (-4,0) -- (6,0) -- (6,3.375) -- (-4,0.5) -- (-4,0);
\fill[Blue,opacity=0.5,path fading=south] (-4,0) -- (6,0) -- (6,-3.375) -- (-4,-0.5) -- (-4,0);
\node at (-7,-1) {\textcolor{Blue}{X-ray}};
%%%%%
\fill[white] (-4,0) ellipse (0.125cm and 0.5cm);
\draw[fill=Green] (-4,0) ellipse (0.125cm and 0.5cm);
\node at (-4,1) {\textcolor{Green}{probe}};
%%%%
\fill[white,opacity=0.5] (-3,-2) rectangle (-2.75,2);
\draw[fill=black,opacity=0.25] (-3,-2) rectangle (-2.75,2);
\foreach \y in {0,1,2,3,4,5,6,7}
{\draw[fill=black,opacity=0.25] (-3,-2+0.5*\y) rectangle (-2.75,-1.75+0.5*\y);
};
\node at (-3,-2.5) {illumination mask};
%\draw[fill=black] (-0.125,-2) rectangle (0.125,-0.25);
%\draw[fill=black] (-0.125,0.25) rectangle (0.125,2);
%%%%
\end{tikzpicture}
	\captionsetup{width=0.85\linewidth}
    \caption{\emph{%Stylized illustration of a diffraction imaging experiment:} 
    Phase retrieval:}
    In diffraction imaging, a probe %(\textcolor{Green}{green}) 
    is illuminated by coherent X-ray light. %(\textcolor{Blue}{blue}). 
    The resulting diffraction pattern is first modulated by an illumination mask %(\textcolor{gray}{gray}) 
    and recorded at detectors in the 2D image plane. Importantly, these detectors can only record intensities, not phases: $y_k = \left|\langle f_k, D^* x \rangle \right|^2$, where $x \in \mathbb{C}^n$ encodes the microscopic structure of the probe, $D=\mathrm{diag}(d_1,\ldots,d_n)$ describes the illumination mask and $f_k \in \mathbb{C}^n$ is a Fourier vector (Fraunhofer approximation to the diffraction equation).
}
    \label{fig:phase-retrieval}
\end{figure}
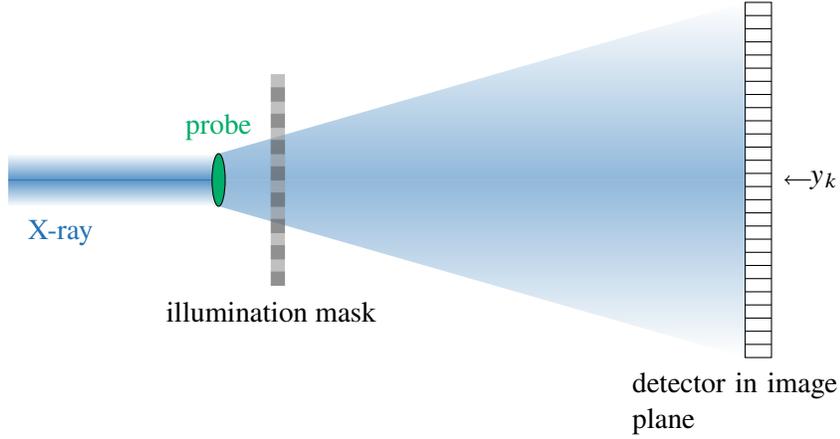

Another instance of a challenging inverse problem is \emph{phase retrieval} -- an important problem with a long history that dates back to the 60s \cite{Wal63}. It occurs naturally in X-ray crystallography \cite{harrison1993phase,millane1990phase}, astronomy \cite{fienup1987phase}, ptychography \cite{rodenburg2008ptychography,HCO+15ptychography}, and quantum tomography \cite{kueng2015orthonormal,kueng2017low}.
We refer to Figure~\ref{fig:phase-retrieval} for a visual illustration.
Mathematically speaking, the discrete phase retrieval problem asks for
inferring a complex signal vector $x \in \mathbb{C}^n$ from $m$
measurements of the form (noiseless for simplicity)
\begin{align}
\tilde{y}_i = \left| \langle a_i,  x_0 \rangle \right|, \quad i\in  [m]. \label{eq:phaseless_measurements1}
\end{align}

This problem %\deleted{ is ill-posed }\fk{is it? I think in the setup we are considering it is not},
%\textcolor{cyan}{ 
cannot be solved unless the measurement system is overcomplete because all phase information is lost in the measurement process. More precisely, it has been shown that one needs $m\geq 4n-4$ generic measurements to ensure that there is a unique solution \cite{CONCA2015346}.

If in contrast, one instead had access to the complex phases $\phi_k$ of $\langle
a_k, x_0 \rangle$, this problem would reduce to solving a linear system of equations:
\begin{align}
\Phi \tilde{y} = A x_0,	\label{eq:phaseless_simplified}
\end{align}
where $\Phi = \sum_{k=1}^m \bar{\phi}_k e_k e_k^*$ and $A = \sum_{k=1}^m e_k a_k^*$ subsumes the measurement process. 
Crucially for phase retrieval, we do not know $\Phi$ in
\eqref{eq:phaseless_simplified}. One intuitive approach to recovering $x_0$ is performing a least-squares minimization over both unknowns:
\begin{align}
\underset{\Phi \in U(m),x \in \mathbb{C}^n}{\textrm{minimize}} \quad \left\| \Phi \tilde{y} - Ax \right\|_2, \label{eq:alternating_minimization}
\end{align}
$\Phi \in U(m)$ is unitary and diagonal in the standard basis.
Although NP hard in general, heuristic approaches exist for solving non-convex problems of this form. One such heuristics is \emph{alternating minimization}, see e.g\ \cite{Fien82,MiaChaKirAy99}.
This is an iterative algorithm, where one alternates between keeping $x$ fixed and minimizing $\Phi$ and, vice-versa: fixing $\Phi$ and optimizing over $x$. Very few theoretical guarantees regarding its performance are known. 

Given the importance of the problem and the lack of mathematical
understanding, obtaining theoretical guarantees for phase retrieval is
highly desirable. In order to do so, we will follow a different
direction pioneered by Balan, Bodmann, Casazza and Eddidin
\cite{BalBodCasEdi09}: lift the quadratic phase retrieval problem to a
linear inverse problem on
positive-semidefinite $n\times n$ matrices:
\begin{align}
y_i = \left| \langle a_i, x_0 \rangle \right|^2 = \mathrm{tr} \langle a_i a_i^*\; X_0 \rangle \quad \text{where} \quad X_0 = x_0x^*_0 \in \mathbb{C}^{n \times n} \label{eq:lifted-measurements}
\end{align}
is proportional to the orthoprojector onto $\mathrm{span}(x_0) \subset \mathbb{C}^n$.
By construction, the desired solution is a Hermitian $n \times n$ matrix with minimal rank ($\mathrm{rk}(X_0)=1$). 
Following Refs.~\cite{CanEldStrVor15,CanStrVor13}, we can exploit this intrinsic rank constraint via constrained nuclear norm minimization~\eqref{SDP_MC}.
This approach effectively reduces the phase retrieval problem to a Hermitian low rank matrix recovery problem, where each linear measurement~\eqref{eq:measurement-operator} must only involve (Hermitian) outer products:
\begin{align*}
y_i = \mathcal{A}\left(x_0x^*_0\right)(i) = \langle A_k, x_0x^*_0 \rangle, \quad \text{where $A_i = a_i a_i^* \in \mathbb{H}_n$ and $i\in[m]$}.
\end{align*}
The reformulation of phase retrieval as a low rank matrix recovery problem has led to the establishment of rigorous recovery guarantees.
By and large, these apply to randomly selected measurement vectors that are sufficiently ``generic''. Exemplary is the main result from Ref.~\cite{candes2014solving}: already $m \gtrsim n$ standard complex Gaussian measurements $a_1,\ldots,a_m \overset{\textit{iid}}{\sim}\mathcal{CN}(0,I)$ suffice to ensure correct recovery. %\textcolor{Blue}{Moreover, this recovery guarantee is stable under noise perturbations.}
Subsequent research has led to similar recovery guarantees for phaseless measurements that are less generic %and, arguably, closer to practice 
\cite{kech15deterministic,candes2015phase,gross2017improved}. We will present two such arguments further below. In Section~\ref{sub:phase-retrieval} we partially derandomize the recovery guarantee for Gaussian measurements by executing a descent-cone analysis. 

We conclude by emphasizing that the phase retrieval problem admits a clean reformulation in terms of low rank matrix recovery. 
This is an ideal starting point for developing rigorous convergence guarantees, but might come with a considerable algorithmic overhead.
After all, we have replaced a non-convex problem over $n$-dimensional vectors by a convex problem over (Hermitian) $n \times n$ matrices~\eqref{SDP_MC}. General purpose solvers, like CVX, quickly run into storage issues as the problem dimension $n$ increases. 
This motivated the development and rigorous analysis of non-convex phase retrieval algorithms.
These include gradient-descent type algorithms on $\mathbb{C}^n$ \cite{candes2015wirtinger,cai2016optimal,truncatedwirtinger}, as well as non-convex approaches based on matrix factorization \cite{BM03Burer-Monteiro,HCO+15ptychography}.
In parallel, the development of matrix sketching algorithms led to substantial storage \& runtime improvements for solving certain convex optimization problems \cite{TYUC17practical-sketching,YUTC17sketchy-decisions}. Importantly, these also apply to lifted phase retrieval and ensure algorithmic tractability even for moderate to large problem sizes \cite{YUTC17sketchy-decisions}.
So, recovery guarantees for lifted phase retrieval -- like the ones presented in this article -- are also of algorithmic relevance.

\subsection{This work}

In this article, we take a look back at more than a decade of rapid progress concerning randomized %ill-posed 
inverse problems for matrix recovery. A complete treatment of all interesting developments would go way beyond the scope of a single article and we choose to focus on one aspect: mathematically rigorous recovery guarantees for the reconstruction of low-rank matrices from generic as well as structured measurements. %generic and structured problem instances.

With the benefit of hindsight, we review two versatile proof techniques and put them into context, namely
the descent cone analysis, as well as the construction of approximate dual certificates.

Section~\ref{sec:descent-cone} deals with the descent cone analysis. That is, low-rank matrix recovery guarantees are obtained by analyzing the relative geometric orientation of the optimization problem's feasible space with respect to the objective function's descent cone anchored at the signal $X_0$ of interest. Exact and unique recovery happens if and only if the intersection of these two convex objects only contains a single point.
Deep results from high-dimensional probability theory

show that this desirable event happens with overwhelming probability, provided that the measurements are sampled independently from sufficiently generic ensembles.
Prominent example applications include optimal generic low-rank matrix recovery (Sub.~\ref{sub:generic-matrix-recovery}), as well as phase retrieval from generic measurement vectors (Sub.~\ref{sub:phase-retrieval}). Although geometrically appealing, this proof technique is not without limitations. It struggles to handle less generic problems, where additional structure -- like incoherence of the unknown signals -- is essential to rule out exceptional problem instances where the reconstruction must necessarily fail. Moreover, this technique does not always give precise insights into the noise-robustness of the reconstruction schemes (Sub.~\ref{sub:limitations}).

Section~\ref{sec:golfing} introduces an alternative proof technique based on duality of convex optimization. 
Convex optimization problems -- like nuclear norm minimization~\eqref{SDP_MC} -- come in pairs and the two problems have a duality gap: objective function values of the primal problem are always smaller or equal to objective function values of the dual problem. Equality occurs if and only if both primal and dual solution are optimal. This, in turn implies, that optimality of a certain feasible point, say $X_0$, can be certified by constructing a dual feasible point that achieves the same objective function value. 
What is more, exact feasibility is not required to certify optimality of $X_0$ for constrained nuclear norm minimization. 
An \emph{approximate dual cerfiticate} suffices, provided that the measurement operator fulfills certain additional properties (Sub.~\ref{sec:dual}).

We will then describe how to construct approximate dual certificates via a probabilistic method -- the so-called golfing scheme (Sub.~\ref{sec:golfing_scheme}). A key advantage of the golfing scheme is that it can be applied to problems with incoherence constraints, where it is not immediately clear how to apply the methods described in Section~\ref{sec:descent-cone}. 
Concrete example applications are matrix completion (Sub.~\ref{sec:app matrix completion}), blind deconvolution and demixing (Sub.~\ref{sec:blind_demixing}), and phase retrieval with incoherence (Sub.~\ref{sec_PR_incoherence}).

Approximate dual certificates do also have their downsides, however. Chief among them is noise robustness. 
In Section~\ref{sec:descent_cone_refined},  
we refine the descent cone arguments introduced in Sub.~\ref{sec:descent-cone} 
to precisely handle noise corruptions.
This leads to near-optimal blind deconvolution guarantees in the high-noise regime (Sub.~\ref{sec:blinddeconv_refined}), as well as novel insights into the phase retrieval problem (Sub.~\ref{sec:phase_incoherence2}).

\section{Recovery guarantees via a descent cone analysis} \label{sec:descent-cone}

\subsection{Descent cone analysis}\label{section:DCA}

Recalling the linear inverse problem \eqref{eqn:measurement model}
$y= \mathcal{A} \left(X_0\right) + e \in \mathbb{C}^m$, there usually is a large set of possible solutions for which $\mathcal{A}$ does not deviate too much from $y$. 
Further properties, such as low rank, can be obtained by minimizing an appropriate function $f: \mathbb{C}^{n_1 \times n_2}\rightarrow \R$ over this set.
If $f$ yields low values only for a small subset of $\{X\in\C^{n_1\times n_2}: \|\mathcal{A}\left(X\right)-y\|_2\le \tau\}$ (or $\{X\in\C^{n_1\times n_2}: \mathcal{A}\left(X\right)=y\}$ in the noiseless case), recovery guarantees can be obtained. This motivates descent cone analysis.

The descent cone $\mathcal{D}(f,X_0)$ of a proper convex function $f: \mathbb{C}^{n_1 \times n_2}\rightarrow \R$ at a point $X_0 \in \mathbb{C}^{n_1 \times n_2}$ is the conic hull of directions in which $f$ decreases near $X_0$: 
\begin{align*}
    \mathcal{D}(f,X_0) :=\{Z\in \mathbb{C}^{n_1 \times n_2}: f(X_0 + \epsilon Z) \le f(X_0) \text{ for some } \epsilon > 0\}.
\end{align*}

Descent cone analysis can facilitate the estimation of probability of success for solving linear inverse problems with optimization.
Consider the following two convex optimization problems (left: noiseless, right: noisy measurements) 
\begin{multicols}{2}
\noindent
\begin{align}\label{OP_nonoise}
\begin{split}
\text{minimize } \quad  & f \left(X\right) \\
\text{subject to} \quad & \mathcal{A} \left(X\right) = y.
\end{split}
\end{align}
\begin{align}\label{OP_noise}
\begin{split}
\text{minimize } \quad  & f \left(X\right) \\
\text{subject to} \quad & \|\mathcal{A} \left(X\right) - y\|_2\le \tau.
\end{split}
\end{align}
\end{multicols}

\begin{figure}[h!]
\centering
\begin{tabular}{ccc}
\begin{tikzpicture}[baseline,scale=0.7]
\begin{scope}[rotate=12]
\draw[thick,Red] (-3,0) -- (3,0); 
%\fill[Blue,opacity=0.5, path fading = north] (-2.75,0) rectangle (2.75,0.5);
\node at (3,0.5) {\textcolor{Red}{$\mathrm{ker}(\mathcal{A})$}};
\fill[black] (0,0) circle(0.075);
\node at (0,0.25) {$0$};
\fill[Blue,opacity=0.25] (0,0) -- (-1.75,-2) -- (1.75,-2) -- (0,0);
\draw[Blue] (0,0) -- (-1.75,-2);
\draw[Blue] (0,0) -- (1.75,-2);
\fill[Blue,opacity=0.25] (1.75,-2) arc(-45:-(180-45):2.5);
\node at (0,-2) {\textcolor{Blue}{$\mathcal{D}(f,X_0)$}};
% \fill[Blue] (0,0) arc (-45:45:3);
\end{scope}
\end{tikzpicture}
& \hspace{1cm} & 
\begin{tikzpicture}[baseline,scale=0.7]
\begin{scope}[rotate=12]
\draw[thick,Red] (-3,0) -- (3,0); 
\draw[Red] (-2.5,0.5) -- (2.5,0.5);
\draw[Red] (-2.5,-0.5) -- (2.5,-0.5);
\fill[Red,opacity=0.125] (-2.5,-0.5) rectangle (2.5,0.5);
%\fill[Blue,opacity=0.5, path fading = north] (-2.75,0) rectangle (2.75,0.5);
\node at (2,1) {\textcolor{Red}{$\| \mathcal{A}(Z)\|_2\leq 2 \tau$}};
%%%
\draw[<->,gray] (-2.5,-0.5) -- (-2.5,0.5);
\node at (-3,0.25) {\textcolor{gray}{$2\delta$}};
\fill[black] (0,0) circle(0.075);
\node at (0,0.25) {$0$};
\fill[Blue,opacity=0.25] (0,0) -- (-1.75,-2) -- (1.75,-2) -- (0,0);
\draw[Blue] (0,0) -- (-1.75,-2);
\draw[Blue] (0,0) -- (1.75,-2);
\fill[Blue,opacity=0.25] (1.75,-2) arc(-45:-(180-45):2.5);
\node at (0,-2) {\textcolor{Blue}{$\mathcal{D}(f,X_0)$}};
% \fill[Blue] (0,0) arc (-45:45:3);
\end{scope}
\end{tikzpicture}
\end{tabular}
\captionsetup{width=0.85\linewidth}
\caption{\emph{Illustration of a descent cone analysis:} The intersection between the nullspace of $\mathcal{A}$ (resp. the set for which $\|\mathcal{A}(Z)\|_2$ is low) with the descent cone $\mathcal{D}(f,X_0)$, i.e., the set of directions $Z$ in which $f$ is decreasing at $X_0$, contains all perturbations $Z$ such that $X_0 + Z$ is a minimizer of the noiseless (resp. noisy) convex optimization problem. (left: noiseless, right: noisy)}
\label{fig:descent-cone}
\end{figure}
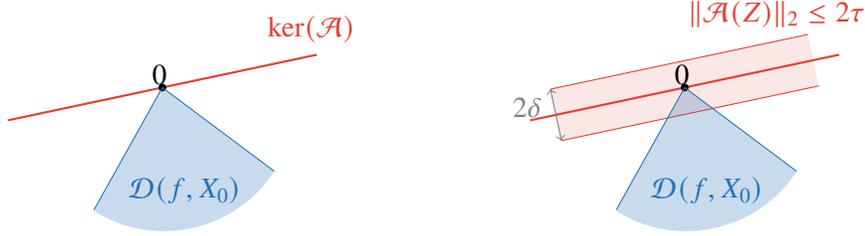

Let us first discuss the noiseless case.
If $X_0$ is the ground truth of the measurements $\mathcal{A} \left(X_0\right)=y$, any minimizer $\hat{X}$ of (\ref{OP_nonoise}) has to fulfill $f(\hat{X})\le f(X_0)$ and $\mathcal{A}\left(\hat{X}\right) = y$, and therefore, can be decomposed as the sum of $X_0$ and a perturbation $Z\in \mathcal{D}(f,X_0)\cap\mathrm{ker}(\mathcal{A})$.
If the intersection between the nullspace $\mathrm{ker}(\mathcal{A})$ and the descent cone $\mathcal{D}(f,X_0)$ only contains the zero element, $X_0$ is the unique optimal solution of (\ref{OP_nonoise}). This is illustrated in Figure~\ref{fig:descent-cone} (left).

This clean geometric picture can be extended to the noisy case.
In this setting, exact recovery cannot be expected. Therefore, we will bound the reconstrution error $\Vert \hat{X} - X_0 \Vert_F = \|Z\|_F$ between a feasible minimizer $\hat{X} = X_0 + Z$ of (\ref{OP_noise}) and the ground truth $X_0$.
Since $\|\mathcal{A} \left(X_0+Z\right) - y\|_2\le \tau$ implies that $\|\mathcal{A} \left(Z\right)\|_2 \le 2\tau$,
the intersection of $\mathcal{D}(f,X_0)$ and $\{Z:\|\mathcal{A} \left(Z\right)\|_2\le 2\tau \}$ has to be analyzed, see Figure~\ref{fig:descent-cone} (right). In order to control the size of this intersection we will need the following quantity, which we refer to as \textit{smallest conic singular value}
   \begin{align*}
    \lambda_{\min} \left( \mathcal{A},  \mathcal{D}(f,X_0)  \right) := \underset{Z \in \mathcal{D}(f,X_0)  \setminus \left\{ 0     \right\}}{\inf} \tfrac{\Vert \mathcal{A} \left(Z\right) \Vert_2}{\Vert Z \Vert_F} \ .
    \end{align*}
If the conic singular value is larger, we expect the intersection to be smaller and, hence, we should obtain stronger noise bound. This intuition is made precise by the following lemma by \cite{chandrasekaran2012convex}, see also \cite{tropp2015convex}.

\begin{lemma}\label{lem:chandrasekaran}\cite[Proposition 2.2]{chandrasekaran2012convex}
	Let $\mathcal A: \mathbb{C}^{n_1 \times n_2}\rightarrow \R^m$ be a linear operator and assume that $y= \mathcal{A} \left(X_0\right) + e $ with $ \Vert e \Vert_2 \le \tau $.
	Then, any minimizer $ \hat{X} $ of the convex optimization problem (\ref{OP_noise}) satisfies
	\begin{align*}
	\Vert \hat{X} - X_0 \Vert_F \le \tfrac{2 \tau}{\lambda_{\min} \left( \mathcal{A},  \mathcal{D}(f,X_0)  \right) }.
	\end{align*}
\end{lemma}

%\begin{proof}[Proof sketch]
\begin{sproof}
By definition, $\lambda_{\min}\left( \mathcal{A},  \mathcal{D}(f,X_0)  \right) \le \tfrac{\Vert \mathcal{A} \left(Z\right) \Vert_{2}}{\Vert Z \Vert_F} \le \tfrac{2\tau}{\Vert Z \Vert_F}$ for any feasible $Z$. The first inequality follows from the definition of $\lambda_{\min}\left( \mathcal{A},  \mathcal{D}(f,X_0)  \right)$ and $Z \in \mathcal{D}(f,X_0)$, the second inequality follows from $\|\mathcal{A}(Z)\|_2\le 2\tau$, which concludes the proof.
\end{sproof}
%\end{proof}

In the following, we will discuss applications with various underlying random operators $\mathcal{A}$. 
We will show how one can obtain lower bounds for the minimum conic singular value, which by Lemma \ref{lem:chandrasekaran} will yield recovery  guarantees, both in the noise-free and in the noisy case.

\subsection{Application 1: generic low rank matrix recovery}
\label{sub:generic-matrix-recovery}

Low-rank matrix recovery describes the problem of recovering a low-rank matrix $X_0 \in \mathbb{R}^{n_1 \times n_2}$ from measurements of the form
\begin{align*}
y_i= \langle A_i, X_0 \rangle \quad \text{where  $A_i \in \R^{n_1\times n_2}$ and $i \in [m]$.}    
\end{align*}
It is useful to introduce the measurement operator $ \mathcal{A}: \R^{n_1 \times n_2}  \rightarrow \R^m$ by
\begin{align}\label{operator:generic_matrix}
    \mathcal{A}(X)(i) := \langle A_i, X \rangle \quad \text{where  $A_i \in \R^{n_1\times n_2}$ and $i\in[m]$.}
\end{align}
for $ X \in \R^{n_1 \times n_2}$. 

(This setting can be extended to the complex-valued setting. However, for simplicity of the exposition, we will only discuss the real-valued setting in this section.)
In this subsection, we will focus on independent random measurement matrices $A_i$ with independent standard normal entries. In order to recover a low-rank matrix $X_0$, we will consider the convex optimization problems (\ref{OP_nonoise}) and (\ref{OP_noise}) with $f=\|\cdot\|_*$.

Recall from the last section that by setting $E:= \{Z\in \mathcal{D}(\|\cdot\|_*,X_0): \|Z\|_F = 1\}$ and bounding $\inf_{Z\in E} \|\mathcal{A}\left(Z\right)\|_2$ the smallest conic singular value from below would guarantee that $\mathcal{D}(\|\cdot\|_*,X_0) \cap \mathrm{ker}(\mathcal{A})$ only contains the zero element and, therefore, exact recovery in the noiseless scenario.

Adjusting Fourcart's and Rauhut's formulation of Gordon's escape through a mesh \cite[Theorem 9.21]{foucart2013mathematical} (originally due to Gordon \cite{Gordon1988}) to the real-valued vector space $\mathbb{R}^{n_1 \times n_2}$, one obtains a powerful lower bound that can exploit the randomness of $\mathcal{A}$.

\begin{theorem}[Gordon's escape through a mesh]\label{thm:Gordon}
Let $\mathcal{A}: \R^{n_1\times n_2}\rightarrow \R^m$ be a Gaussian measurement operator as defined in (\ref{operator:generic_matrix}) and let $E$ be a subset of the Frobenius unit sphere $S_F(\R^{n_1\times n_2}):= \{Z\in \R^{n_1\times n_2}: \|Z\|_F=1\}$. Further, define the Gaussian width of $E$ as 
\begin{align}
    \ell(E) := \mathbb{E}\sup_{Z\in E} \langle A, Z\rangle, \label{eq:gaussian-width}
\end{align}
where $A \in \mathbb{R}^{n_1 \times n_2}$ is a standard normal matrix ($A_{ij} \overset{\textit{iid}}{\sim} \mathcal{N}(0,1)$).
Then, for t>0:
    \begin{align*}
        \inf_{Z\in E}\|\mathcal{A}(Z)\|_2 \ge \sqrt{m-1} - \ell(E)-t
    \end{align*}
with probability at least $1-e^{-t^2/2}$.
\end{theorem}

The Gaussian width is actually a reasonable summary parameter for the size of a convex cone. It is also closely related to the statistical dimension \cite{amelunxen2014}. If $\ell(E)$ does not exceed $\sqrt{m-1}$ recovery guarantees can be obtained.

Theorem \ref{thm:Gordon} only requires $E$ to be a subset of the Frobenius unit sphere and, therefore, one is not restricted to a specific descent cone but one can instead choose the union over all possible descent cones corresponding to rank-$r$ matrices in order to obtain uniform recovery guarantees:
\begin{align*}
E_r = S_F(\R^{n_1\times n_2}) \cap K_r \quad \text{and} \quad K_r = \bigcup_{X \in \R^{n_1\times n_2}: \mathrm{rk}(X)=r} \mathcal{D} \left( \| \cdot \|_*, X \right).
\end{align*}

H\"older's inequality yields $\sup_{Z\in E_r} \langle A, Z\rangle \le \|A\| \sup_{Z\in E_r}\|Z\|_*$. 
Tight bounds on the operator norm of a standard Gaussian matrix are readily available (more on that later), but it seems plausible that the largest nuclear norm of $Z \in E_r$ could scale unfavorably with the ambient dimension ($\|Z\|_* \leq \sqrt{\min \left\{n_1,n_2\right\}} \|Z\|_F$ which is sharp). The geometry of descent cones, however, excludes such worst-case instances.
The following lemma highlights that the effective rank of descent cone elements is proportional to the rank of the anchor point. It is a generalization of \cite[Lemma~10]{kueng2017low} to rectangular matrices.
To increase accessability, we write $x \lesssim y$ if there is a postive constant $C>0$ such that $x \leq Cy$.

\begin{lemma}\label{lem:effective-rank}
Suppose that $Z \in \mathbb{C}^{n_1 \times n_2}$  is contained in the nuclear norm descent cone of a rank-$r$ matrix $X \in \mathbb{C}^{n_1 \times n_2}$. Then,
\begin{align*}
\|Z\|_* %\leq \left(1+\sqrt{2}\right) 
\lesssim \sqrt{r} \|Z\|_F.
\end{align*}
\end{lemma}

The suppressed proportionality constant is small ($C \leq 1+\sqrt{2}$), but probably not optimal.
The proof is novel and uses ideas from dual certificates (Section~\ref{sec:golfing}), as well as pinching, see e.g.\ \cite[Problem~II.5.4]{bhatia}. We refer to Appendix~\ref{sec:effective-rank} for details. With this lemma at hand, we can bound the Gaussian width of $E_r$.

\begin{corollary}\label{col:matrixrec}
    The Gaussian width of $E_r$, the union over all possible descent cones with an anchor point of rank-$r$ can be bounded by
    \begin{align*}
        \ell(E_r) %\le  \left(1+\sqrt{2}\right) \sqrt{r} (\sqrt{n_1}+\sqrt{n_2})
        \lesssim \sqrt{r} \left( \sqrt{n_1} + \sqrt{n_2}\right).
    \end{align*}
    Further, let $\mathcal{A}: \R^{n_1\times n_2}\rightarrow \R^m$ be a Gaussian measurement operator as defined in (\ref{operator:generic_matrix}). Then, $\lambda_{\min} \left( \mathcal{A},  \mathcal{D}(f,X)\right)$ is bounded away from zero for any rank-$r$ matrix X w.h.p. if
    \begin{align*}
        m \gtrsim  r (n_1 + n_2).
    \end{align*}
\end{corollary}

\begin{sproof}
    Using H\"older's inequality and Lemma \ref{lem:effective-rank}, the Gaussian width $\ell(E_r)$ can be bounded in terms of the expected operator norm of a standard Gaussian matrix:
\begin{align*}
    \ell(E_r) = \mathbb{E}\sup_{Z\in E_r} \langle A, Z\rangle \leq \sup_{Z \in E_r} \|Z \|_* \,\mathbb{E} \|A \| \lesssim %\left(1+\sqrt{2}\right) 
    \sqrt{r} \mathbb{E}\|A\|.
    %\le \left(1+\sqrt{2}\right) \sqrt{r} (\sqrt{n_1}+\sqrt{n_2})
\end{align*}
A tight upper bound $\mathbb{E}\|A\| \le(\sqrt{n_1}+\sqrt{n_2})$ can be found, e.g. in \cite[p.292]{foucart2013mathematical}. By Theorem \ref{thm:Gordon},
\begin{align*}
    \inf_{X\in E_r}\|\mathcal{A}(Z)\|_2 = \inf_{X\in\R^{n_1\times n_2}: \mathrm{rk}(X)=r} \lambda_{\min} \left( \mathcal{A},  \mathcal{D}(\Vert \cdot \Vert_{\ast},X_0)  \right) \ge \sqrt{m-1} - \ell(E) - t
\end{align*}
 with probability at least $1-e^{-t^2/2}$. Therefore, if
%\begin{equation*}
$
    m \gtrsim r (n_1+n_2)% \ge 1 + (1+\sqrt{2})^2r (n_1 + n_2 + 2\sqrt{n_1 n_2})
$.
%\end{equation*}.
we can pick $t>0$, such that $\inf_{X\in E_r}\|\mathcal{A}(Z)\|_2 $ is positive w.h.p.
\end{sproof}

Even when measuring multiple matrices of rank-$r$ via the same measurement operator $\mathcal{A}$, Corollary \ref{col:matrixrec} uniformly bounds $\lambda_{\min} \left( \mathcal{A},  \mathcal{D}(\Vert \cdot \Vert_{\ast},X)  \right)$ from below and, therefore, gives a uniform recovery guarantee for recovering not only one but all possible rank-$r$ matrices. %signals.

\subsection{Application 2: phase retrieval} \label{sub:phase-retrieval}

Recall that $\mathbb{H}_n \subset \mathbb{C}^{n \times n}$ denotes the (real-valued) vector space of Hermitian $n \times n$ matrices.
The lifted reformulation of the phase retrieval problem is based on the measurement operator
\begin{align*}
\mathcal{A}(X_0)(i) = \langle A_i, X_0 \rangle \quad A_i = a_i a_i^* \in \mathbb{H}_n, \; X_0 = x_0 x_0^* \in \mathbb{H}_n, \; i \in \left[m\right]. 
\end{align*}
This bears strong similarities with the measurement operator for generic low rank matrix recovery~\eqref{operator:generic_matrix}, but there is one crucial distinction. Each measurement matrix $A_i = a_i a_i^*$ is itself a rank-one orthoprojector. These are everything but generic random matrices (c.f.\ a matrix with standard normal entries is almost surely \emph{not} rank-deficient) and a clean descent cone analysis based on Gordon's escape through a mesh (Theorem~\ref{thm:Gordon}) seems out of reach.
Fortunately, Mendelson and co-authors \cite{mendelson2014learning,koltchinskii2015bounding} developed a weaker variant of Theorem~\ref{thm:Gordon}. Known as Mendelson's small ball method, this result only requires i.i.d.\ measurement matrices that also obey a small ball property. 
%It also replaced the Gaussian width by an empirical width parameter. 
We refer to Tropp \cite{tropp2015convex} for a user-friendly exposition and proof and state it directly in terms of measurement operators on Hermitian $n \times n$ matrices. 

\begin{theorem}[Mendelson's small ball method] \label{thm:mendelson}
Suppose that $\mathcal{A}:\mathbb{H}_n \to \mathbb{R}^m$ is a measurement operator~\eqref{eq:measurement-operator} whose measurements correspond to independent realizations of a Hermitian random matrix $A \in \mathbb{H}_n$.
Fix a subset $E \subset \mathbb{H}_n$ and for $\xi >0$ define
\begin{align*}
Q_\xi (E; A)=& \inf_{Y \in E} \mathrm{Pr} \left[ \left| \langle A, Y \rangle \right| \geq \xi \right],\\
W_m (E;A) =& \mathbb{E} \sup_{Y \in E} \langle Y,H \rangle \quad H = \tfrac{1}{\sqrt{m}}\sum_{i=1}^m \epsilon_i A_i,
\end{align*}
where $\epsilon_1,\ldots,\epsilon_m \overset{\textit{iid}}{\sim}\left\{\pm 1 \right\}$ is a Rademacher sequence. Then, for any $\xi>0$ and $t>0$,
\begin{align}
\inf_{Y \in E} \| \mathcal{A}(Y) \|_2 \geq \xi \sqrt{m} Q_{2\xi}(E;\Phi) - 2 W_m (E;\Phi)-\xi t \label{eq:mendelson}
\end{align}
with probability at least $1-\mathrm{e}^{-2t^2}$.
\end{theorem}

In fact, this statement is valid for all real-valued\footnote{Extensions to complex-valued inner product spaces are also possible, see e.g.\ \cite{jung2019combinatorial}.} inner product spaces with finite dimensions.
It is worthwhile to point out that for
standard normal random matrices $\Phi_1,\ldots,\Phi_m \in \mathbb{R}^{n_1 \times n_2}$ and subsets $E$ of the Frobenius unit sphere, this
result recovers Theorem~\ref{thm:Gordon} up to constants. Fix $\xi >0$ of appropriate size. Then, $E \subset \left\{Y \in \mathbb{H}_n:\; \|Y \|_2=1\right\}$ ensures that $\xi Q_{2 \xi}(A;E)$ is constant. What is more, $W_m (A,E)$ reduces to the usual Gaussian width~\eqref{eq:gaussian-width}.

We obtain a recovery guarantee for phase retrieval by appropriately analyzing both contributions to Eq.~\eqref{eq:mendelson}.
Similar to before, we can actually obtain a uniform recovery guarantee by taking into account all possible descent cones in one go:
\begin{align}
E_1 = \left\{Y \in \mathbb{H}_n:\; \|Y\|_F=1\right\} \cap K_1, \quad \text{where} \quad K_1 = \bigcup_{x \in \mathbb{C}^n} \mathcal{D} \left( \| \cdot \|_*, x x^* \right).
\label{eq:phase-retrieval-descent-cone}
\end{align}

Let us start with controlling the empirical width.

\begin{lemma}[empirical width for non-generic phase retrieval] \label{lem:phase-retrieval-width}
Let $E_1 \subset \mathbb{H}_n$ be the union of descent cones defined in Eq.~\eqref{eq:phase-retrieval-descent-cone} and
suppose that $a \in \mathbb{C}^n$ is an isotropic, sub-normalized random vector, i.e.\  $\mathbb{E} a a^* = \Id$, $\|a \|_2 \leq \sqrt{2 n}$. 
Then, %\textcolor{orange}{for $m \lesssim n \log (n)$}
\begin{align}
W_m (E_1) \lesssim \sqrt{n \log (n)} \quad \text{provided that $m \lesssim n \log (n)$.}\label{eq:phaseless-width}
\end{align}
\end{lemma}

The assumption $m \lesssim n \log (n)$ is not essential, but will simplify exposition later on.
%The proportionality constant is modest (smaller than four).
Similar arguments apply to standard complex Gaussian measurement vectors $g \in \mathbb{C}^n$ (which are not sub-normalized) and produce tighter bounds  \cite{kueng2017low}: $W_m (E_1,aa^*) \lesssim \sqrt{n}$ (no $\log(n)$-factor), provided that $m \lesssim n$. The following proof sketch summarizes arguments presented in Ref.~\cite{kueng2017low}.

\begin{sproof}[Lemma~\ref{lem:phase-retrieval-width}]
We will show the slightly more general bound 
\begin{align*}
\mathbb{E} \|H \| \lesssim \sqrt{\max \left\{m, n \log (n) \right\}}.
\end{align*}
Apply Lemma~\ref{lem:effective-rank} to obtain
\begin{align*}
W_m (E_1,A) = \mathbb{E} \sup_{Y \in E_1} \langle Y, H \rangle 
\lesssim \sup_{Y \in E_1} \|Y\|_* \mathbb{E} \| H \| \leq %1+\sqrt{3})
\sqrt{r} \mathbb{E} \|H\|.
\end{align*}
The remaining expression is an operator norm of a random matrix $H=\tfrac{1}{\sqrt{m}}\sum_{i=1}^m \epsilon_i a_i a_i^*$ that features two types of randomness. 
The matrix Khintchine inequality, see e.g.\ \cite[Exercise 8.6(d)]{foucart2013mathematical} allows us to trade the Rademacher randomness against an additional square root. More precisely, 
\begin{align*}
\mathbb{E} \|H \|=
\mathbb{E}_{a} \mathbb{E}_{\epsilon} \| H \| %=&  \mathbb{E}_{a} \tfrac{1}{\sqrt{m}} \mathbb{E}_{\epsilon} \Big\| \sum_{j=1}^m \epsilon_j a_j a_j^* \Big\|_\infty
\lesssim  \mathbb{E}_{a} \sqrt{\tfrac{\log (n)}{m}} \Big\| \sum_{j=1}^m (a_j a_j^*)^2 \Big\|^{1/2} 
\lesssim  \sqrt{\tfrac{ n \log (n)}{m}}  \mathbb{E}_{a} \Big\| \sum_{j=1}^m a_j a_j^* \Big\|^{1/2}
\end{align*}
where the last inequality follows from $(a_j a_j^*) (a_j a_j^*) = \|a_j\|_2 \ a_j a_j^*\lesssim \sqrt{n} \ a_j a_j^*$. We now face an operator norm of a sum of  random matrices $X_j=a_j a_j^*$ that are positive semidefinite and obey $\|X_i \| = \|a_j \|_2^2 \leq 2n$ each. Isotropy also asserts $\Big\| \sum_{j=1}^m \mathbb{E} X_j \Big\| = \| m \Id \| =m$ and we can apply the matrix Chernoff inequality \cite{tropp2012user-friendly} to obtain for any $\tau >0$
\begin{align*}
\mathbb{E}_{a} \Big\| \sum_{j=1}^m a_j a_j^* \Big\| \leq \tfrac{\mathrm{e}^\tau-1}{\tau}m + \tfrac{\sqrt{2}}{\tau} n \log (n) %\leq \tfrac{1}{\tau}\left( \mathrm{e}^\tau-1+\sqrt{2}\right) \max \left\{ m,n \log (n)\right\}
\lesssim  \max \left\{m,n \log (n) \right\}.
\end{align*}
{}
\end{sproof}

The empirical width bound~\eqref{eq:phaseless-width} suggests that an order of $n \log (n)$ non-generic phaseless measurements may suffice to establish strong uniform recovery guarantees for phase retrieval via low rank matrix reconstruction. However, this is only true if the measurement matrices $a_i a_i^*$ are not too spikey. More precisely, we need that $Q_{2\xi}(aa^*,E_1)$ -- the second quantity in Mendelson's small ball method~\eqref{eq:mendelson} -- is lower-bounded by a constant. 

\begin{lemma}[marginal tail function for non-generic phase retrieval] \label{lem:marginal}
Suppose $a\in\mathbb{C}^n$ is a random vector that obeys 
$\mathbb{E} \langle a, Y a \rangle^2 \gtrsim \langle Y,Y \rangle$ and 
$\mathbb{E} \langle a, Y a \rangle^4 \lesssim \left(\mathbb{E} \langle a, Y a \rangle^2\right)^2$ for all $Y \in E_1$. Then, 
\begin{align*}
Q_{2\xi}(E_1;aa^*) \gtrsim \left(1-\tfrac{4 \xi^2}{\mathrm{const}}\right)^2
\quad \text{for all} \quad 0 < \xi < \sqrt{\mathrm{const}/4}.
\end{align*}
\end{lemma}

\begin{proof}
Fix $Y \in E_1$ and use $\mathbb{E} \langle a, Z a \rangle^2 \gtrsim \langle Y,Y\rangle = \mathrm{const}$ to apply a Paley-Zygmund type argument:
\begin{align*}
\mathrm{Pr} \left[ \left| \langle aa^*, Y \rangle \right| \geq 2 \xi \right] \geq \mathrm{Pr} \left[ \langle a, Y a \rangle^2 \geq \tfrac{4 \xi^2}{\mathrm{const}} \mathbb{E} \langle a, Y a \rangle^2 \right]
\geq \left( 1-  \tfrac{4\xi^2}{\mathrm{const}} \right)^2 \tfrac{\left(\mathbb{E} \langle a, Z a \rangle^2\right)^2}{\mathbb{E} \langle a,Z,a \rangle^4}.
\end{align*}
The moment assumption $\mathbb{E} \langle a, Y a \rangle^4 \lesssim \left( \mathbb{E} \langle a, Y a \rangle^2 \right)^2$ ensures that the final ratio is lower-bounded by a constant. Such a lower bound is valid, regardless of $Y \in E_r$. Hence, it also applies to the infimum $Q_{2\xi} = \inf_{Y \in E_1} \mathrm{Pr} \left[ | \langle aa^*,Y \rangle | \geq 2 \xi \right]$.
\end{proof}

We now have gathered all the auxiliary statements we need to carry out a descent-cone analysis for phase retrieval with non-generic measurements.

\begin{theorem}[phase retrieval from non-generic measurements] \label{thm:phaseless}
Let $a\in \mathbb{C}^n$ be a random vector that is isotropic ($\mathbb{E} aa^*=\Id$), sub-normalized ($\|a \|_2 \leq \sqrt{2n}$) and also obeys
\begin{align}
\mathbb{E} \langle a,Y a \rangle^2  \gtrsim \langle Y,Y \rangle, %\geq \mathrm{const} >0, 
\quad \text{as well as} \quad \left( \mathbb{E} \langle a,Y a \rangle^2 \right)^2 \lesssim \mathbb{E} \langle a,Y A \rangle^4, \label{eq:phaseless-assumption}
\end{align}
for every $Y \in K_1$.
Then, with high probability, a total of 
\begin{align*}
m \gtrsim  n \log (n)
\end{align*}
randomly selected phaseless measurements $a_1,\ldots,a_m \sim a \in \mathbb{C}^n$ suffice to reconstruct signals $x_0 \in \mathbb{C}^n$ via constrained nuclear norm minimization~\eqref{SDP_MC}. 
\end{theorem}

In fact, this recovery guarantee is actually \emph{uniform}. That is, with high probability a single collection of randomly sampled phaseless measurements allows for reconstructing \emph{all} phaseless signals via nuclear norm minimization \eqref{SDP_MC}. Conditioned on this event, the actual reconstruction is also stable with respect to noise corruption. 
Suppose that $y = \mathcal{A}(xx^*)+e$, where $\|e \|_{\ell_2} \leq
\tau$ and the noise bound is known. Then, the solution $\hat{X}$ of
the convex optimization problem~\eqref{SDP_MC} is guaranteed to obey
$\|\hat{X} - x_0{x_0}^* \|_F \lesssim \tau/\sqrt{m}$.
Up to constants, this assertion is on par with some of the strongest stablility guarantees for low rank matrix reconstruction in general \cite{candes2011probabilistic,candes2014solving,kabanava2016stable}.

\begin{sproof}[Theorem~\ref{thm:phaseless}]
Let us start by reformulating phase retrieval as a low rank matrix recovery problem ($r=1$). The general descent cone analysis presented in Section~\ref{section:DCA} identifies the minimum conic singular value as an important summary parameter. If it is positive, the current set of measurements allows to recover $X_0=x_0x_0^*$ via nuclear norm minimization under idealized circumstances (no noise). The size of the minimum conic singular value also captures noise robustness (the larger the better). Theorem~\ref{thm:mendelson} (Mendelson's small ball method) achieves just that. Fix $\xi = \mathrm{const}$ sufficiently small and insert the bounds from Lemma~\ref{lem:marginal} and Lemma~\ref{lem:phase-retrieval-width} into the assertion of Theorem~\ref{thm:mendelson}:
\begin{align*}
\inf_{Y \in E_1} \| \mathcal{A}(Y) \|_2 \geq & \xi \sqrt{m} Q_{2\xi} \left(E_1;aa^*\right) - 2 W_m \left(E_1; aa^*\right) - \xi t \\
\gtrsim & \sqrt{m} - \mathrm{const} \left(\sqrt{n \log (n)} + t\right),
\end{align*}
with probability at least $1- \mathrm{e}^{-2t^2}$. Assigning $m = C n \log (n)$ and $t = \gamma/2 \sqrt{m}$, where $C>0$ ($\gamma>0$) is a sufficienlty large (small) constant, allows us to conclude $\inf_{Y \in E_1} \| \mathcal{A} (Y) \|_2 \gtrsim \sqrt{m}$ with probability at least $1-\mathrm{e}^{-\gamma \sqrt{m}}$. This ensures that the minimum conic singular value is of (optimal) order $\sqrt{m}$.

There is one additional twist. In Eq.~\eqref{eq:phase-retrieval-descent-cone} we have defined the set $E_1$ as the union of all possible descent cones anchored at all possible lifted signals $X=xx^*$. Consequently, Theorem~\ref{thm:mendelson} produces a lower bound of $\sqrt{m}$ on the infimum over all possible descent cones, not just a single one. This allows us to effectively treat all possible signals at once and establish a uniform recovery guarantee.
\end{sproof}

Let us conclude this section with discussing the extra assumptions~\eqref{eq:phaseless-assumption}. They formulate conditions on the second- and fourth moment of the measurement matrices $A=aa^*$. The second moment condition ensures that the expected measurement operator is non-singular on the union $K_1$ of all descent cones: %(intersected with the unit sphere):
\begin{align}
\tfrac{1}{m}\mathbb{E} \langle Y, \mathbb{E} \mathcal{A}^*\mathcal{A} (Y) \rangle = \mathbb{E} \langle aa^*,Y \rangle_F = \mathbb{E} \langle a,Y a \rangle^2 \gtrsim \langle Y,Y \rangle 
\quad \text{for all $Y \in K_1$.} \label{eq:lifted-isotropy}
\end{align}
Viewed from this angle, it actually captures (sub-)isotropy on the relevant parts of $\mathbb{H}_n$ -- a natural requirement for any low rank matrix recovery procedure.
Alas, by itself it is not sufficient to derive nontrivial recovery guarantees \cite{kueng2015spherical,gross2015partial} and extra assumptions are required.
Theorem~\ref{thm:phaseless}, for instance, requires that (certain) fourth moments of $A=aa^*$ are comparable to their second moment squared. It should be viewed as a relaxation of (sub-)Gaussian moment growth conditions, but only up to order four. Suitable measurement ensembles only need to mimic (outer products of) Gaussian measurement vectors up to 4th moments. This condition is much weaker than subgaussianity, and vector distributions that satisfy Eq.~\eqref{eq:phaseless-assumption} can admit a lot of structure. 
A concrete example are orbits of certain symplectic symmetry groups that arise naturally in quantum information (Clifford group) and time-frequency analysis (oscillator group) \cite{kueng2016low}.
A more refined analysis also allows for replacing constrained nuclear norm minimization~\eqref{SDP_MC} by a simple least-squares or $\ell_p$-fit over the cone of positive semidefinite matrices \cite{kabanava2016stable}, such as the convex optimization problem
\begin{align}\label{opt:noSDP}
\begin{split}
\text{minimize } \quad & \sum_{i=1}^{m} \big\vert \trace \left( \xi^{\left(i\right)} (\xi^{\left(i\right)})^* X  \right) -y_i \big\vert\\
\text{ subject to} \quad & X \in \mathcal{S}^n_{+},
\end{split}
\end{align}
where $ \mathcal{S}^n \subset \R^{n \times n} $ denotes the set of real-valued symmetric matrices and $\mathcal{S}^n_{+} \subseteq \mathcal{S}^n$ its positive definite subset.
%We will present one such optimization problem below in Eq.~\eqref{opt:SDP}.
Such reformulations have the added benefit of being tuning-free. In
particular, no a priori noise bound $\tau$ is required, see
\cite{kueng2018nonnegative} for related arguments addressing sparse
vector recovery and \cite{Fengler:TIT:Bayesian:2021} sparse covariance
matching.

\subsection{Limitations}\label{sub:limitations}

As we have seen, a descent-cone analysis combined with probabilistic tools such as Mendelson's small ball method yields essentially near-optimal uniform recovery results for low-rank matrix recovery from Gaussian measurement matrices or phase retrieval measurements with Gaussian measurement vectors. 
A key observation of the proof is that the union of all descent sets is contained in a suitably large nuclear norm ball, so it suffices to estimate the Gaussian width of this ball. 

This approach, however, has significant limitations when it comes to problems with more structure such as matrix completion and  blind deconvolution.
The reason is that in these problems, as explained in Section~\ref{subsub:matrixcompletion}, recovery guarantees will necessarily fail for some exceptional signals that violate certain incoherence conditions.
Thus it will necessarily be impossible to bound the minimum conic singular values for the descent cones anchored at these signals and estimating a general superset cannot be sufficient.

However, one may wonder whether it is possible to obtain a comparable result to Corollary \ref{col:matrixrec} and Theorem \ref{thm:phaseless} by considering the \textit{the union of all descent cones of all incoherent rank-$r$ matrices} instead. However, this turns out to be more delicate. In particular, it is unclear how to mathematically formulate a property that captures the fact that matrices in the descent cone anchored at incoherent signals are better conditioned with respect to the measurements. 
A direct connection to the notion of incoherence is difficult, as matrices in the descent cone anchored at coherent signals will not necessarily be incoherent. As a consequence, also the minimum conic singular values can become provably very small \cite{asilomar2018,krahmer2019convex}, which makes it difficult to bound them from below, which would be necessary for recovery guarantees based on the strategy explained above even for the noiseless case. 

In the next section, we present an alternative analysis strategy that is better suited to deal with incoherence conditions, as it is based on (approximate) dual certificates rather than the descent cone and relies on the signal alone rather than differences to alternative solutions. In certain cases, however, as we will see in Section \ref{sec:refined}, it will also be possible to adapt the descent cone analysis to such scenarios.

\section{Recovery guarantees via the golfing scheme} \label{sec:golfing}

\subsection{Recovery guarantees via dual certificates}
\label{sec:dual}

Maybe the most natural way of proving that a convex optimization attains its optimal value at a given argument is by exhibiting a \emph{dual certificate} 
-- the generalization to possibly non-smooth convex functions of the familiar gradient condition for optimality.
Let us start by considering the noiseless nuclear norm problem ($\tau=0$)
\begin{align}
\begin{split}
\underset{X \in \mathbb{C}^{n_1 \times n_2}}{\text{minimize}} & \quad \| X \|_* \label{SDP3}\\
\text{subject to }& \quad \mathcal{A}(X) = y, 
\end{split}
\end{align}
see also Eq.~\eqref{OP_nonoise} with $f(X) = \|X \|_*$.
Let $X\in \mathbb{C}^{n_1 \times n_2}$ be a rank-$r$ matrix with singular value decomposition (SVD) $ X= U \Sigma V^* $. That is,
 $ \Sigma \in \mathbb{R}^{r \times r}  $ is a diagonal matrix with nonnegative entries and $U \in \mathbb{C}^{ n_1 \times r }$ and $V \in \mathbb{C}^{n_2 \times r}$ are isometries, i.e., $  U^*U  = V^* V = \Id_r$. 
The tangent space of the variety of rank-$r$ matrices at the point $X$ can be checked to be given by
\begin{align}\label{equ:deftangentspace}
  T_X := \left\{  UA^* + B V^*    : \ A \in \mathbb{C}^{n_2 \times r}, B \in \mathbb{C}^{n_1 \times r}     \right\}.
\end{align}
Denote by 
$ \mathcal{P}_{T_X}$ 
the (Hilbert-Schmidt) orthogonal projection onto the tangent space, and by
$\mathcal{P}_{T^{\perp}_X}$ the projection onto its ortho-complement.
The \emph{subdifferential} 
$\partial \Vert \cdot \Vert_{\ast} \left( X \right)$
of the nuclear norm at $X$ is the set of affine lower-bounds to the nuclear norm that coincide with the norm at $X$.
A simple application of the matrix H\"older inequality \cite{bhatia} shows that
\cite{watson_subdifferential}
\begin{align}\label{equ:charsubdifferential}
\partial \Vert \cdot \Vert_{\ast} \left( X \right) = \left\{  W \in \mathbb{C}^{n_1 \times n_2} :  \ \mathcal{P}_{T_X} W = UV^*, \   \Vert \mathcal{P}_{T^{\perp}_X} W \Vert\le 1         \right\}.
\end{align}
With these notions, it is straight-forward to see that a sufficient condition for $X_0$ being the minimizer of \eqref{SDP3} is given by the following lemma, first formulated in Ref.~\cite{candes2009exact}.

\begin{lemma}[\cite{candes2009exact}]\label{lemma:candesrechtmatrix}
	Let $X_0 \in \mathbb{C}^{n_1 \times n_2}$ be such that $\mathcal{A}\left( X_0 \right) = y \in \mathbb{C}^m$. 
	Suppose that the following two conditions hold:
	\begin{enumerate}
		\item There exists a vector $z \in \mathbb{C}^m$ such that $Y=\mathcal{A}^* \left(z\right) $ satisfies
		\begin{align*}
		\mathcal{P}_{T_{X_0}}  Y = UV^* \quad \text{and} \quad \Vert \mathcal{P}_{T^{\perp}_{X_0}} Y  \Vert < 1.
		\end{align*}
		\item The linear operator $ \mathcal{A}$ is injective when restricted to the tangent space $T_{X_0}$.
	\end{enumerate}
	Then $X_0$ is the unique minimizer of \eqref{SDP3}.
\end{lemma}

In Ref.~\cite{candes2009exact} it was shown in the context of low-rank matrix completion from a sufficient number of uniformly sampled matrix elements, that such a dual certificate exists with high probability.
A refined (and fairly involved) analysis in Ref.~\cite{candes2010power} showed that the number of measurements can be reduced to the order of the information-theoretic limit, up to logarithmic factors. 

Reference~\cite{gross2011recovering} introduced a new approach -- the \emph{golfing scheme} -- for constructing dual certificates.
In the original paper, and commonly in works referring to it, the result is presented as being based on the observation that the conditions in Lemma~\ref{lemma:candesrechtmatrix} can be relaxed and that the existence of an \emph{approximate dual certificate} suffices to establish uniqueness.
Approximate dual certificates are easier to construct using randomized processes, which in their natural formulations will give results that are correct only approximately and up to a small probability of failure.

In this article, we aim to present the story from a different point of view.
Namely we will show that a minor tweak of the golfing scheme actually gives an explicit randomized construction for an \emph{exact} dual certificate, using no more measurements than the original argument.
In this sense, it is inaccurate to say that constructing exact certificates is harder than constructing approximate ones.
While this point of view does not seem to impact the headline result on recovery guarantees, we feel that it represents a conceptually clearer way of thinking about the argument.
To the best of our knowledge, this approach has not appeared elsewhere in the literature before.

\subsection{Golfing with precision}\label{sec:golfing_scheme}

Here, we recall the basic logic behind the golfing scheme, in preparation of presenting the \emph{putting proposition}, Proposition~\ref{prop:putt}.
We start with two definitions.

The analysis uses the fact that the measurement operator $\mathcal{A}$ is an approximate isometry when restricted to the subspace $T_X$.
The precise notion employed is this:
\begin{definition}\label{deltarip}
	Let $ X \in \mathbb{C}^{n_1 \times n_2}$. We say that $\mathcal{A}$ fulfills the $\delta$-restricted isometry property ($\delta$-RIP) on $T_{X}$,
	%T_X is defined above. If one understands the def, the inclusion is obvious. If not, the inclusion doesn't help.
%	\textcolor{orange}{\subset \mathbb{C}^{n_1 \times n_2}}$,
	if for all matrices $Z \in T_{X} $ it holds that
	\begin{align*}
	\left(1-\delta\right) \Vert Z \Vert^2_F \le \Vert \mathcal{A} \left(Z\right) \Vert^2_2 \le \left(1+\delta\right) \Vert Z \Vert^2_F.
	\end{align*}
\end{definition}

As the name suggests, approximate dual certificates obey condition~1 in Lemma~\ref{lemma:candesrechtmatrix} approximately. This is captured by the following formal definition.

\begin{definition} \label{def:approx-dual}
  Given a measurement operator $\mathcal{A}: \mathbb{C}^{n_1\times n_2} \to \mathbb{C}^m$, 
  a vector $z\in\C^m$, giving rise to a matrix $Y=\mathcal{A}^* \left(z\right)$,
  is an \emph{approximate dual certificate}  at $X_0=U\Sigma V^*$ if it satisfies the following properties:
  \begin{align}
	&\Vert z \Vert_2 \le 2, 
	\label{eq:approx1}
	\\
	&\alpha = \Vert UV^* -  \mathcal{P}_{T_{X_0}} \mathcal{A}^*(z)\Vert_F \le \tfrac{1}{8 \Vert \mathcal{A} \Vert}, \label{eq:approx1.5}
	\\
	&\Big\Vert \mathcal{P}_{T^{\perp}_{X_0}} \left( \mathcal{A}^* \left(z\right) \right) \Big\Vert < \tfrac{1}{2}.
	\label{eq:approx3}
  \end{align}
\end{definition}

With these definitions, the central result reads:

\begin{proposition}\label{mainprop}\cite{gross2011recovering,candes2010matrix} Let $X_0 \in \C^{n_1 \times n_2}$ with SVD $X_0 = U\Sigma V^* $ and suppose that $y=\mathcal{A} \left(X_0\right) + e$ with $ \Vert e \Vert_2 \le \tau $. Suppose that the following two conditions hold.
	\begin{enumerate}
		\item There exists an
		  approximate dual certificate $Y=\mathcal{A}^*(z)$
		  % i mean, that's directly above.
		  %\textcolor{orange}{(see Definition~\ref{def:approx-dual})}.
%		  vector $z \in \mathbb{C}^{n_1 \times n_2}$ with $\Vert z \Vert \le 2 $ such that the matrix $Y=\mathcal{A}^* \left(z\right)$ satisfies
%		\begin{equation}\label{ineq:intern2}
%		\alpha :=\Vert UV^* -  \mathcal{P}_{T^{\perp}_{X_0}} Y   \Vert_F< \tfrac{1}{8 \Vert \mathcal{A} \Vert} \text{ and } \Vert \mathcal{P}_{T^{\perp}_{X_0}} Y  \Vert < \tfrac{1}{2}
%		\end{equation}
		\item The measurement operator $\mathcal{A}$ satisfies the $\delta$-restricted isometry property on $T_{X_0}$ with constant $\delta=3/4$
		  %\textcolor{orange}{(See Definition~\ref{deltarip})}
	\end{enumerate}
	Then, every minimizer $ \hat{X}$ of \eqref{SDP3} satisfies
	\begin{align}\label{eq:errorbound}
	\Vert X_0 - \hat{X} \Vert_F \lesssim \Vert \mathcal{A} \Vert \tau.
	\end{align}
\end{proposition}

Here, $\|\mathcal A\| = \sup_{\|z\|_2 = 1} \|\mathcal{A}(z)\|_F$ is the operator norm of the measurement operator.
The requirement (\ref{eq:approx1}) on the norm of $z$ is only necessary in the noisy case $\|e\|_2>0$. 

The bound~\eqref{eq:errorbound} is not always tight. 
For example, let $\mathcal{A} $ be the Gaussian measurement operator defined in Section \ref{sub:generic-matrix-recovery}. 
For $m \ll n_1 n_2$, $ \Vert \mathcal{A}\Vert \asymp \sqrt{n_1 n_2}$ with high probability. 
This is larger than the optimal error scaling $\Vert X_0 - \hat{X} \Vert_F \propto \sqrt{m}$ for this regime and measurement model.

\begin{figure}
    \centering
\begin{tabular}{ccccc}
\begin{tikzpicture}[baseline,scale=0.7]
\coordinate (Y0) at (0,0);
\coordinate (Y1) at (0.5,2);
\coordinate (Y2) at (2.25,2);
\coordinate (Y3) at (2.5,3);
\coordinate (Yid) at (3,3);
%%%%%%%%%
\node[anchor=south] at (0,3) {\large (1)};
\fill[black] (Y0) circle (0.0625cm);
\node[anchor=north] at (Y0) {$Y_0$}; 
\fill[black] (Y1) circle (0.0625cm);
\node[anchor=east] at (Y1) {$Y_1$};
%\fill[black] (Y2) circle (0.0625cm);
%\node[anchor=west] at (Y2) {$Y_2$};
\fill[black] (Yid) circle (0.0625cm);
\node[anchor=south west] at (Yid) {$Y_{\mathrm{id}}$}; 
%%%%%
\draw[thick,dashed,->] (Y0) -- (Yid);
\draw[thick,->] (Y0) -- (Y1);
%\draw[thick,dashed,->] (Y1) -- (Yid);
%\draw[thick,->] (Y1) -- (Y2);
%\draw[thick,->] (Y2) -- (Y3);
%\draw[thick,dotted,->] (Y3) -- (Yid);
\end{tikzpicture}
& \hspace{1cm} &
\begin{tikzpicture}[baseline,scale=0.7]
\coordinate (Y0) at (0,0);
\coordinate (Y1) at (0.5,2);
\coordinate (Y2) at (2.25,2);
\coordinate (Y3) at (2.5,3);
\coordinate (Yid) at (3,3);
%%%%%%%%%
\node[anchor=south] at (0,3) {\large (2)};
\fill[black] (Y0) circle (0.0625cm);
\node[anchor=north] at (Y0) {$Y_0$}; 
\fill[black] (Y1) circle (0.0625cm);
\node[anchor=east] at (Y1) {$Y_1$};
\fill[black] (Y2) circle (0.0625cm);
\node[anchor=west] at (Y2) {$Y_2$};
\fill[black] (Yid) circle (0.0625cm);
\node[anchor=south west] at (Yid) {$Y_{\mathrm{id}}$}; 
%%%%%
%\draw[thick,dashed,->] (Y0) -- (Yid);
%\draw[thick,->] (Y0) -- (Y1);
\draw[thick,dashed,->] (Y1) -- (Yid);
\draw[thick,->] (Y1) -- (Y2);
%\draw[thick,->] (Y2) -- (Y3);
%\draw[thick,dotted,->] (Y3) -- (Yid);
\end{tikzpicture}
& \hspace{1cm} & 
\begin{tikzpicture}[baseline,scale=0.7]
\coordinate (Y0) at (0,0);
\coordinate (Y1) at (0.5,2);
\coordinate (Y2) at (2.25,2);
\coordinate (Y3) at (2.5,3);
\coordinate (Yid) at (3,3);
%%%%%%%%%
\node[anchor=south] at (0,3) {\large (3)};
\fill[black] (Y0) circle (0.0625cm);
\node[anchor=north] at (Y0) {$Y_0$}; 
\fill[black] (Y1) circle (0.0625cm);
\node[anchor=east] at (Y1) {$Y_1$};
\fill[black] (Y2) circle (0.0625cm);
\node[anchor=west] at (Y2) {$Y_2$};
\fill[black] (Yid) circle (0.0625cm);
\node[anchor=south west] at (Yid) {$Y_{\mathrm{id}}$}; 
%%%%%
%\draw[thick,dashed,->] (Y0) -- (Yid);
\draw[thick,->] (Y0) -- (Y1);
%\draw[thick,dashed,->] (Y1) -- (Yid);
\draw[thick,->] (Y1) -- (Y2);
\draw[thick,->] (Y2) -- (Y3);
\draw[thick,dotted,->] (Y3) -- (Yid);
\end{tikzpicture}
\end{tabular}
\captionsetup{width=0.85\linewidth}
\caption{
  \emph{Construction of an approximate dual certificate via golfing:}
  For the $p$-th leg, we start with a current best guess $Y_p$ (c.f.~panels~(1), (2) above).
  On the tangent space $T_{X_0}$, 
  we aim to express the difference $\Delta_p:=Y_{\mathrm{id}} - \mathcal{P}_{T_{X_o}} Y_p$ (dotted line) in terms of rows from the partial measurement matrix $\mathcal{A}^p$.
  Here, $Y_{\mathrm{id}}=UV^*$ is the ``ideal'' dual certificate, which is an element of the tangent space.
	If the rows of $\mathcal{A}^p$ were an orthonormal basis, then $\Delta_p = (\mathcal{A}^p)^* \mathcal{A}^p(Y_{\mathrm{id}})$ would give an exact solution.
  If $\mathcal{A}^p$ is subsampled from an ortho-normal basis, standard measure concentration results imply that on the tangent space, we will obtain a relatively decent approximation for $\Delta_p$ (solid line).
  In fact, if the number of rows in $\mathcal{A}^p$ is sufficient, one can easily show that the distance to the ideal certificate will be reduced by a constant factor with high probability.
  It is then natural to just iterate the scheme (panel (3)).
  This results in a random process which converges in Frobenius norm to the ideal certificate (on the tangent space) exponentially quickly.
  At the same time, on the space orthogonal to the tangent space, we have that $\mathbb{E}\big( \mathcal{P}_{T^{\perp}_{X_0}} \left( \mathcal{A}^* \left(z\right) \right) \big)=0$.
  Again using concentration of measure results, one can show that during the logarithmically many legs of the golfing procedure, the spectral norml of these terms remains small.
}
\label{fig:golfing}
\end{figure}
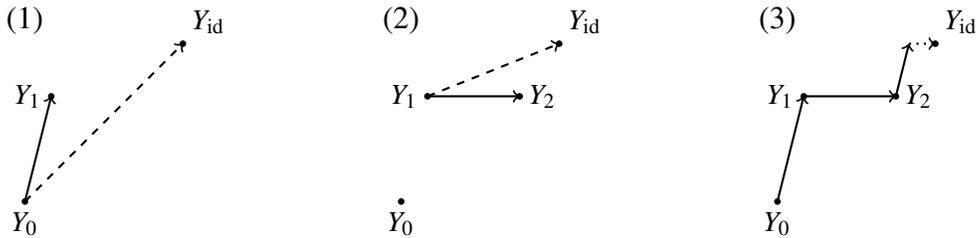

Before proving this statement,
we sketch the idea behind the \emph{golfing scheme} \cite{gross2011recovering} for the construction of an approximate dual certificate (c.f.~Figure~\ref{fig:golfing}). 

The ensemble of measurement vectors will often be \emph{isotropic} in the sense that $ \mathbb{E} \left[\mathcal{A}^* \mathcal{A}\right] = \Id $. 
This motivates the choice $\tilde{z}_1= \mathcal{A} \left(UV^*\right) $ and $\tilde{Y}_1 = \mathcal{A}^* \left(z\right) = \mathcal{A}^* \mathcal{A} \left(UV^*\right) $ for $z$ and $Y$, as it leads to the correct result
$ \mathbb{E} \left[ \mathcal{A}^* (\tilde{z}_1) \right]= UV^* $ \emph{in expectation}.
Consequently, one could then hope to show properties \eqref{eq:approx1}, \eqref{eq:approx1.5}, and \eqref{eq:approx3} using measure concentration around the mean. 
Unfortunately, this approach does not usually work directly.
One problem is that the operator norm $\Vert \mathcal{A} \Vert$
can be quite large (for blind deconvolution $\Vert \mathcal{A}  \Vert$ it is of the order $\sqrt{KN/L}$. 
This, in turn, means that $\Big\Vert UV^* - \mathcal{P}_{T_{X_0}} \mathcal{A}^* \left(z\right) \Big\Vert_F $ needs to be small, smaller than typical fluctuations.
The idea behind the golfing scheme is to iteratively refine this initial guess until condition \eqref{eq:approx1.5} is satisfied:
\begin{itemize}
	\item \textbf{Step 1:} Choose a partition of $\left[m\right]$ into $Q$ disjoint sets 
$\left\{\Gamma_1,\ldots,\Gamma_q \right\}$ of size roughly $\vert \Gamma_q \vert \approx m/Q $, such that $ Q \mathbb{E} \left[ (\mathcal{A}^q)^* \mathcal{A}^q \right]   \approx \Id$, where $\mathcal{A}^q := Q_{\Gamma_q} \mathcal{A} $. (Here, $Q_{\Gamma_q}: \mathbb{C}^m \rightarrow \mathbb{C}^m$ denotes the coordinate projection onto $\Gamma_q$.)
	\item \textbf{Step 2:} Set 
	\begin{align*}
	Y_0&=0 \quad \text{and }\\
	Y_{q}&= Y_{q-1} + Q_{\Gamma_q}   (\mathcal{A}^q)^* \mathcal{A}^q \left(UV^* - \mathcal{P}_{T_{X_0}} Y_{q-1} \right) \quad \text{ where $1 \leq q \leq Q$}.
	\end{align*}
	\noindent The corresponding $z\in \mathbb{C}^m$ is then given by 
	\begin{align*}
	z:= Q \sum_{q=1}^{Q}  \mathcal{A}^q \left( UV^* - \mathcal{P}_{T_{X_0}} Y_{q-1}  \right).
	\end{align*}
\end{itemize}
Note that a consequence of the sample splitting in Step 1 is that the golfing scheme is set up in a such a way that the distribution of $ \mathcal{A}^q $ is independent of $Y_{q-1}$.
This simplifies the analysis, but is not essential \cite{gross2010note}.

The precise convergence properties of this random process depend on the parameters (partition size, incoherence, etc. \cite{gross2011recovering}) and is, in any case, beyond the scope of this article.
Instead, we want to make precise the following new observation -- which, in keeping with the theme, we call the \emph{putting proposition}.
For more context, see the discussion at the end of Section~\ref{sec:dual}.

\begin{proposition}[``Putting Proposition'']\label{prop:putt}
  Assume that the approximate dual certificate properties (\ref{eq:approx1}) -- (\ref{eq:approx3}) hold and that $\mathcal{A}$ fulfills  the $\delta$-restricted isometry property on $T_{X_0}$ for $\delta<3/4$.
  Then there exists an exact dual certificate for $X_0$.
\end{proposition}

\begin{proof}
  Using the  variational characterization of the operator norm of a Hermitian linear map, as well as the definition of the $\delta$-RIP, we get 
  \begin{align*}
	\left\|
	  P_{T_{X_0}} 
	  \mathcal{A}^* \mathcal{A}
	  P_{T_{X_0}} 
	  -
	  P_{T_{X_0}} 
	\right\|
	&=
	\sup_{Z\in{T_{X_0}}, \|Z\|_F=1}
	\big|
	(Z,\mathcal{A}^* \mathcal{A} Z)
	  -
	 1  
	\big|
	 \\
	&=
	\sup_{Z\in{T_{X_0}}, \|Z\|_F=1}
	\big|
	\|\mathcal{A}Z\|_F^2
	  -
	 1 
	\big|
	\leq
	\delta.
  \end{align*}
  Hence, as a linear map on the tanget space,
  $P_{T_{X_0}} \mathcal{A}^* \mathcal{A} P_{T_{X_0}}$
  is invertible and satisfies
  \begin{align*}
	\big\| \big(P_{T_{X_0}} \mathcal{A}^* \mathcal{A} P_{T_{X_0}} \big)^{-1} \big\| 
	\leq
	\tfrac{1}{1-\delta}.
  \end{align*}

  Set
  \begin{align*}
	x 
	= 
	\big(\mathcal{A} P_{T_{X_0}} \big)^{-1} 
	\big(
	  UV^* -  
	  \mathcal{P}_{T_{X_0}} 
	  \mathcal{A}^*(z) 
	\big).
  \end{align*}
  Together with (\ref{eq:approx1.5}), this gives
  \begin{align*}
	\|x\|_2
	\leq 
	\tfrac{1}{\sqrt{1-\delta}}
	\big\| UV^* -  \mathcal{P}_{T_{X_0}} \mathcal{A}^*(z) \big\|_F
	\leq
	\tfrac{1}{8 \sqrt{1-\delta} \|\mathcal{A}\|}.
  \end{align*}
  But then, with $Y' = \mathcal{A}^*(z-x) = Y - \mathcal{A}^*(x) $, we have that
  \begin{align*}
	\mathcal{P}_{T_{X_0}}(Y') = UV^*
  \end{align*}
  and
  \begin{align*}
	\left\|
	  \mathcal{P}_{T_{X_0}^\perp}Y' 
	\right\|
	&\leq
	\left\|
	  \mathcal{P}_{T_{X_0}^\perp}Y 
	\right\|
	+
	\left\|
	\mathcal{P}_{T_{X_0}^\perp}\mathcal{A}^*(x)
	\right\| \\
	&\leq
	\tfrac12
	+
	\left\|
	\mathcal{P}_{T_{X_0}^\perp}\mathcal{A}^*(x)
	\right\|_F \\
	&\leq
	\tfrac12
	+
	\|\mathcal{A}\| \|x\|_2 \\
	&\leq
	\tfrac12
	+
	\tfrac{1}{8\sqrt{1-\delta}}
	<1.
  \end{align*}
{}
\end{proof}

\subsection{Application 3: matrix completion}
\label{sec:app matrix completion}

Using dual certificate-based proof techniques, the nuclear norm minimization
approach to matrix completion has been studied extensively
\cite{candes2009exact,candes2010power,gross2011recovering,recht2011simpler,chen_coherenceoptimal}.
A typical result for the noiseless case ($\tau=0$) reads as follows:
\begin{theorem}[\cite{chen_coherenceoptimal}]\label{theorem:grossMC}Assume that $n_1 \ge n_2$. Consider measurements of the form $ y = \mathcal{A} \left(X_0\right)$, where $X_0\in \mathbb{R}^{n_1 \times n_2}$  is a rank-$r$ matrix and $\mathcal{A}$ is given by (\ref{equ:MCopdefinition}). Assume  that
	\begin{align*}
	m \ge C  \max \left\{ \mu^2 \left(U\right), \mu^2 \left(V\right)   \right\}  r n_1 \log^2 n_1.
	\end{align*}
	Then with high probability,  the matrix $X_0$ is the unique minimizer of 
	SDP~\eqref{SDP3} (see also SDP~\eqref{SDP_MC} with $\tau=0$).
\end{theorem}

Further variants have been studied in the literature.
For example, if the noise term is drawn randomly instead of adversarially, improved results can be given, see Refs.~\cite{koltchinskii2011nuclear,klopp} for subexponential, and 
Ref.~\cite{chenconvex} for subgaussian noise.
Non-convex algorithms with rigorous performance guarantees can be found in Refs.~\cite{irlsfornasier,irlskuemmerle,optspace1,Jain_alternatingminimization,sun2016guaranteed,ma2017implicit, ge2016matrix,optspace2}).

\subsection{Application 4: simultaneous demixing and blind deconvolution}\label{sec:blind_demixing}

Simultaneous blind deconvolution and demixing is a generalization of the blind deconvolution problem introduced in Section~\ref{section:blinddeconv}.
It is motivated by wireless communication scenarios that involve multiple senders, but only one receiver. Each sender wants to transmit a signal $m_i$ using a linear encoder $C_i$. The encoded signal $x_i=C_{i} m_i$ is sent through an unknown convolution channel $w_i$ to the receiver. Because there are multiple senders, the receiver obtains the superposition of $r$ convolutions, where the goal is to reconstruct all messages $\left\{ m_i \right\}_{i=1}^r$.

In mathematical terms, this leads to an inverse problem of the form
\begin{align}\label{ydef}
y= \sum_{i=1}^{r} w_i \ast x_i +e \in \mathbb{C}^L,
\end{align}
where $*$ denotes the (circular) convolution introduced in Eq.~\eqref{eq:convolution}. The goal is to simultaneously reconstruct all signals $x_i$, as well as all channel descriptions $w_i$.
 As in the randomized blind deconvolution framework, we have to use some prior knowledge on $w_i$ and $x_i$ in order to be able to reconstruct these signals. We are going to adopt the framework introduced in Ref.~\cite{ling2017blind}.
Assume that $w_i$ and $x_i$ are elements of known subspaces. Hence, we can write $w_i = B h_i$ and $x_i= C_i m_i$ for all $ i \in \left[r\right] $, where $B \in \mathbb{C}^{L \times K} $ and $C_i \in \C^{L \times N} $. We assume that $B^* B = \Id$ and, moreover, that for each $i \in \left[ r \right] $ the entries of the matrix $C_i$ are i.i.d. samples from the complex normal distribution $\mathcal{CN} \left(0,1\right)$.
\begin{figure}[ht]
	\centering
	\includegraphics[width=.5\linewidth]{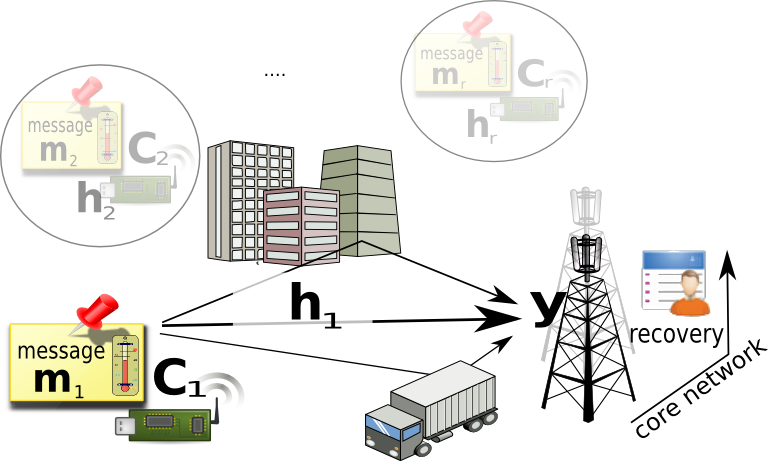}
	\captionsetup{width=0.85\linewidth}
	\caption{\emph{A multi-user wireless (uplink) communication scenario:}
          wireless devices $i=1\dots r$ simultaneously transmit
          messages $m_i$ to the basestation which are individually
          encoded with a linear code $C_i$ and experience individual
          convolutional channels $h_i$.}
	\label{fig1}
\end{figure}

Similar to the the randomized blind deconvolution setting, we note that for each $i \in \left[ r \right]$ there is a unique linear operator $\mathcal{A}_i: \mathbb{C}^{K \times N} \rightarrow \mathbb{C}^L$ such that for all $ u \in \C^{K}$ and $ v \in \C^{N}$ it holds that 
\begin{align*}
\mathcal{A}_{i} \left(uv^*\right)=B u\ast C_i \overline{v}.
\end{align*}
Hence, we can rewrite equation \eqref{ydef} as
\begin{align*}
y= \sum_{i=1}^{r} \mathcal{A}_i \left( h_i m^*_i \right) +e.
\end{align*}
We emphasize that each outer product $h_i m_i^*$ comes with a unique linear operator $\mathcal{A}_i$. This allows us to recast Eq.~\eqref{ydef} as a low-rank matrix recovery problem on a larger space. 
Our goal is to recover the block-diagonal rank-$r$ matrix
\begin{align*}
  X_0= h_1 m_1^* \oplus h_2 m_2^* \oplus \cdots \oplus h_r m_r^* %, h_2 m^*_2, \ldots, h_r m^*_r \right).
\end{align*}
from a linear measurement operator that decomposes accordingly ($\mathcal{A}(
Z_1 \oplus \cdots \oplus Z_r) = \sum_{i=1}^r \mathcal{A}_i (Z_i)$). 
Adapting SDP~\eqref{SDP_MC} to this problem structure yields
\begin{align}\label{ineq:help5}
\begin{split}
\underset{X_1,\ldots,X_r \in \mathbb{C}^{K \times N}}{\text{minimize}} \quad & \sum_{i=1}^{r} \Vert X_i \Vert_{\ast}\\
\text{subject to} \quad  &\Vert y -  \sum_{i=1}^{r}  \mathcal{A}_i \left(X_i\right) \Vert_2 \le \tau,\end{split}
\end{align}
see \cite{ling2017blind}.
Furthermore, denote by $\mu^2_{\max} $ and $\mu^2_{h}  $ the coherence parameters, which are similar to the ones defined in Section \ref{section:blinddeconv}. (For a precise definition we refer to \cite{jung2018blind}.)
In \cite{ling2017blind} it has been shown that if
\begin{align}\label{ineq:help7}
L \gtrsim r^2 \left(K\mu^2_{\max}+N \mu^2_{h} \right) \log^3 L
\end{align}
holds, then in the noiseless scenario, i.e., $e=0$, the convex relaxation \eqref{ineq:help5} recovers the ground truth matrix $X_0$ with high probability. 

However, we observe that the number of degrees of freedom in this problem is $ r \left(K + N -1\right) $, which raises the question, whether the quadratic dependence on $r$ in \eqref{ineq:help7} is necessary. Indeed, numerical experiments in \cite{ling2017blind} indicate that the true dependence of the sample complexity in $r$ should rather be linear (see \cite[Section IV]{ling2017blind} as well as \cite[Section III]{jung2018blind}).

The main result in \cite{jung2018blind} shows that the required simple complexity is indeed linear in $r$. Hence, nuclear norm-minimization can recover the ground truth signal $X_0$ at near-optimal sample complexity.

\begin{theorem}(\cite{jung2018blind}, see also \cite{stoger2016blind,stoger2017blindspie,stoger2017blind}) \label{theorem:mainwithnoise}
	Let $y \in \mathbb{C}^L$ be given by (\ref{ydef}) with $ \Vert e \Vert_2 \le \tau $. Assume that 
	\begin{align*}
	L/\log^3 L \gtrsim r \left(  K \mu^2_{\max} \log \left( K \mu^2_{\max} \right) + N \mu^2_{h} \right).
	\end{align*}
	Then, with high probability, every minimizer
	$ \hat{\mathsf{X}} =  \hat{X}_1 \oplus \ldots \oplus \hat{X}_r $ of SDP \eqref{ineq:help5} satisfies
	\begin{align*}
	\sqrt{ \sum_{i=1}^{r} \Big\Vert \hat{X}_i  - h_i m^*_i \Big\Vert^2_{F}} \lesssim \tau  \sqrt{ rN }  .
	\end{align*}
\end{theorem}

In the following, we are going to describe the main technical ingredient, which allowed for linear scaling in $r$. In both \cite{ling2017blind} and \cite{jung2018blind}, the proofs establish the existence of an approximate dual certificate with high probability. One ingredient is to show that the measurement operator acts as an approximate isometry operator on the tangent space of $X_0$, see Definition~\ref{deltarip}. To make this precise in the blind demixing scenario, define, for $ i \in \left[r\right] $, the tangent space $T_i$ of rank-$1$-matrices at $h_im^*_i$:
\begin{align*}
T_i = \left\{ h_i u_i^* + v_i m^*_i : \ u_i \in \mathbb{C}^{K}, v_i \in \mathbb{C}^{N}   \right\}
\end{align*}
Then we can define the tangent space at $X_0$ by
\begin{align*}
\tilde{T} := \left\{  X_1 \oplus \ldots \oplus X_r: \ X_i \in T_i \quad \text{for all } i \in \lbrack r \rbrack   \right\}.
\end{align*} 
%This allows us to give the following definition of a $\delta$-local isometry, which is analogous to Definition \ref{deltarip}.
In both \cite{ling2017blind} and \cite{jung2018blind}, one part of the proof consists in showing that, with high probability, the collection of measurement operators $\left\{ \mathcal{A}_i  \right\}^r_{i=1} $ fulfill a local isometry property on $\tilde{T}$. That is, for a sufficiently small $ \delta > 0 $
	\begin{align}\label{ineq:deltalocalisometryoutline}
	\left(1 - \delta \right)  \Vert X \Vert^2_F \le \Big\Vert \sum^r_{i=1} \mathcal{A}_i \left( X_i\right) \Big\Vert_2^2 \le \left(1 + \delta \right) \Vert X \Vert^2_F \quad \text{for all $X = X_1 \oplus \ldots \oplus X_r \in \tilde{T}$.}
	\end{align}
%Using this definition one can state the sufficient conditions for recovery, see \cite[Lemma 18]{jung2018blind}. 
In   \cite{ling2017blind}, the restricted isometry property is first shown individually on each $T_i$ and after that is shown that the images of the subspaces $T_i$ under the operator $\mathcal{A}_i$ are sufficiently near-orthogonal to each other. Combining these two properties yields \eqref{ineq:deltalocalisometryoutline}. However, the second step requires that $L$ scales quadratically in $r$.

In contrast, our analysis establishes the restricted isometry property directly on $\hat{T}$. For that, we define 
\begin{align*}
\hat{T} := \left\{  X_1 \oplus \ldots \oplus X_r \in \tilde{T}: \sum_{i=1}^r \Vert X_i \Vert_F^2 =1    \right\}.
\end{align*}
Next, we observe that \eqref{ineq:deltalocalisometryoutline} is equivalent to
\begin{align*}
\delta & \ge \underset{ X_1 \oplus \ldots \oplus X_r \in \hat{T} }{\sup} \Big\vert \ \Big\Vert \sum_{i=1}^{r} \mathcal{A}_i \left(X_i\right) \Big\Vert^2_{2}  - \sum_{i=1}^{r} \Vert X_i \Vert^2_F  \Big\vert \\
&=  \underset{ X_1 \oplus \ldots \oplus X_r \in \hat{T}  }{\sup} \Big\vert \  \Big\Vert \sum_{i=1}^{r} \mathcal{A}_i \left(X_i\right) \Big\Vert^2_{2}  -   \mathbb{E} \Big\lbrack \Big\Vert \sum_{i=1}^{r} \mathcal{A}_i \left(X_i\right) \Big\Vert^2_{2} \Big\rbrack     \Big\vert.
\end{align*}
The key idea is that the last expression can be interpreted as a suprema of chaos processes and we can use deep results from empirical process theory \cite{krahmer2014suprema} to bound this expression with high probability.

\subsection{Phase retrieval with incoherence}\label{sec_PR_incoherence}
Recall that in the phase retrieval problem we are interested in reconstructing a signal $x_0$ from measurements of the form
\begin{align}\label{phase_subgaussian_measurements}
y_k = \vert \langle a_k, x_0 \rangle \vert^2 +e_k.    
\end{align}
We have seen in Section \ref{sub:phase-retrieval} that this problem can be solved not only for Gaussian measurement vectors $\left\{ a_i \right\}$, but also for  measurement vectors that are less generic. Nevertheless, the required assumptions are somewhat more restrictive than for example in compressive sensing. In particular, for measurement vectors with unimodular entries
%\deleted{for certain  very simple distributions such as random measurement vectors consisting of i.i.d.~random signs }
the problem does not even have a unique solution. %\deleted{In fact, measurement vectors with unimodular entries are never suitable.

To see that, assume that for all $k$ the entries of the vector $a_k$ have all the same modulus, i.e.
\begin{align}\label{equ:same_mod}
\vert \left( a_k \right)_1 \vert = \vert \left( a_k \right)_2 \vert = \ldots = \vert \left( a_k \right)_n \vert. 
\end{align}
In this case, both the vectors
\begin{align}\label{phase_retrieval_counterexample}
\begin{split}
x_1 &:= \left( 1,0, \ldots, 0 \right) \in \R^n \\   
x_2 &:= \left( 0,1, \ldots, 0 \right) \in \R^n  
\end{split}
\end{align}
lead to the same measurements, i.e.
\begin{align*}
 \vert \langle a_{i}, x_1 \rangle \vert^2 = \vert \langle a_{i}, x_2 \rangle \vert^2 \quad \text{for all $ i \in \left[ m \right]$.}
\end{align*}
 Hence, $x_1$ and $x_2$ cannot be distinguished based on phaseless measurements alone. We want to stress that condition \eqref{equ:same_mod} holds for several interesting classes of measurement vectors. For example, this condition is fulfilled if the entries $a_k$ are Rademacher random variables, i.e., $\left(a_k \right)_i$ is either $1$ or $-1$, each with probability $1/2$. Moreover, if each entry $\left(a_k \right)_i$ is a random variable with uniform distribution over $ S^1 \subset \C $ this condition would also be fulfilled. 

Another example, which is important for certain applications, is given by random masks \cite{candes2015phase,gross2017improved}. That is, the measurement vector $a_k$ are of the form
\begin{align*}
a_k = \diag \left( \epsilon_k \right) f_{l_k},    
\end{align*}
where $\epsilon_k \in \left\{ -1, 1 \right\}^{n}$ is a Rademacher vector and $f_{l_k}$ is the $l_k$-th column of the DFT matrix $F\in \C^{n \times n}$.

A first step to address these issues was taken in \cite{krahmer2018phase}. The key idea is to  impose an incoherence condition of the form
\begin{align}\label{phase:incoherence} 
\tfrac{\Vert x_0 \Vert_{\infty}}{\Vert x_0 \Vert_2} \le \mu < 1,
\end{align}
which prevents counterexamples of the form \eqref{phase_retrieval_counterexample}. Under such incoherence condition, one can obtain recovery guarantees for all centered random vectors with i.i.d. real-valued subgaussian entries of unit variance, including the case of Rademacher random vectors that was previously excluded. More precisely, \cite[Theorem V.1]{krahmer2018phase} yields that with high probability, all signals satisfying \eqref{phase:incoherence} for $\mu=\tfrac{1}{\sqrt{8}}$ can be recovered via \eqref{opt:noSDP} from an order-optimal number of measurements. The proof combines the golfing scheme with stability bounds of \cite{eldar_mendelson}, confirming that the golfing scheme is well-suited to deal with incoherence.
 
We note that that the incoherence condition \eqref{phase:incoherence} is much weaker than the incoherence conditions in matrix completion and blind deconvolution because it is dimension-free.  %\deleted{However,  note that }\textcolor{cyan}{
At the same time, this approach is limited to the real case, as the underlying stability results from \cite{eldar_mendelson} exploit that the phase factors to be recovered are actually signs and hence belong to a finite candidate set. Thus for the complex case, one needs different tools, which will be discussed in Section~\ref{sec:phase_incoherence2} below.

\section{More refined descent cone analysis}\label{sec:descent_cone_refined}
\label{sec:refined}
\subsection{Application 5: blind deconvolution}\label{sec:blinddeconv_refined}
In Section \ref{sub:limitations}, we have discussed why the descent cone analysis framework described in Section \ref{section:DCA} cannot be directly applied to the matrix completion and blind deconvolution scenario. In the following, we want to outline how one can refine those methods to obtain novel insights into low-rank matrix recovery problems. For that, we are going to revisit the blind deconvolution setting, see Section \ref{section:blinddeconv}, and demonstrate how to combine a descent cone analysis with incoherence constraints to prove near-optimal bounds in settings which are relevant in practice. This improves over existing error bounds (see, e.g. \cite{ahmed2014blind}), which depend polynomially on $K$ and $N$, and hence are quite pessimistic.
%where we are going to show that these refinements allow to improve the noise-bound in Theorem ??? and one can avoid the dimension-dependent factor $ \sqrt{K+N}$.\\

More precisely, recall from Section \ref{sub:limitations} that we only expect to obtain reasonable bounds for matrices with low incoherence, as described by the set
\begin{align*}
\mathcal{H_{\mu}}:= \left\{ h_0 \in \C^K: \ \ \sqrt{L} \vert \innerproduct{ b_{\ell},h_0 } \vert \le \mu \Vert h_0 \Vert_2 \text{ for all } \ell \in \left[L\right]   \right\}.
\end{align*}

Even if the signal is contained in the set, not all principal components of a descent direction need to be incoherent as well. The key observation underlying the following theorem is that these ``coherent'' descent directions only allow for very small decrements and will hence only play a significant role for very small noise levels. Thus even under mild lower bounds on the noise level, one obtains near-optimal recovery guarantees.

\begin{theorem}\cite[Theorem 3.7]{krahmer2019convex}\label{thm:stabilityBD}
	Let $ \alpha>0 $ and $ B \in \C^{L\times K} $ such that $B^* B= \Id $. Assume that
	\begin{align*}
	L \gtrsim   \tfrac{\mu^2}{\alpha^2}  \left(K+N\right) \ \log^2 L.
	\end{align*}
	Then with high probability the following statement holds for all $ h_0 \in \mathcal{H}_{\mu} \setminus \left\{ 0 \right\} $, all $m_0 \in \C^N \setminus \left\{ 0 \right\} $, all $ \tau > 0 $, and all $ e \in \C^L $ with $ \Vert e \Vert_2 \le \tau$ :\\
	Any $\hat{X}$ minimizing the nuclear norm subject to a data fidelity term of at most $\tau$ satisfies
	\begin{align}\label{equ:noisebound}
	\Vert \hat{X} - h_0 m^*_0 \Vert_F \lesssim \tfrac{  \mu^{2/3}  \log^{2/3} L}{\alpha^{2/3}} \max \left\{  \tau  , \alpha \Vert h_0 m^*_0 \Vert_F  \right\}.
	\end{align}
\end{theorem}
Note that the error estimate in \eqref{equ:noisebound} depends only logarithmically on $L$.  
To illustrate this result, assume that the noise level $\tau = \epsilon \mu^{-2} \log^{-2} L$ for some $\epsilon>\epsilon_0$. Then, by setting $\alpha \asymp \epsilon_0 \mu^{-2} \log^{-2} L$, we obtain near-linear error bounds with a required sample complexity at the order of
	\begin{align*}
L \ge C_1  \tfrac{\mu^6}{\epsilon_0^2}  \left(K+N\right) \ \log^6 L.
\end{align*}
This improves over existing noise bounds as in \cite{ahmed2014blind} and shows that for large enough noise near-optimal recovery bounds are possible.
\begin{sproof}
As discussed in Section \ref{sub:limitations}, the minimum conic singular value of the descent cone at the point $h_0 m_0^* $ is ill-conditioned, i.e., there exists a matrix $Z\in \C^{K \times N}$ such that $ \tfrac{\Vert \mathcal{A} \left(Z\right) \Vert_2}{\Vert Z \Vert_F}$ is small. The key observation in the proof is that only matrices $Z$, which are near-orthogonal to the ground truth, can be poorly conditioned. This observation gives rise to the following proof strategy. Namely, we partition the descent cone of the nuclear norm at the point $h_0 m_0^*$ into two cones $\mathcal{K}_1$ and $\mathcal{K}_2$, where the cone $\mathcal{K}_1$ contains all the directions, which are almost orthogonal to the ground truth matrix $h_0 m^*_0 $. The cone $\mathcal{K}_2 $ contains all the remaining directions.
It turns out that matrices in the descent cone $\mathcal{K}_2$ inherit certain coherence properties from the matrices $h_0 m_0^*$, which allows us to apply Mendelson's small ball method to obtain a lower bound for the minimum conic singular value $ \lambda_{\min} \left( \mathcal{A},  \mathcal{K}_2\right) $, which is at the order of a constant (up to $\log$-factors and ignoring the $\mu$-dependence). Then, using Lemma~%\deleted{\eqref{lem:chandrasekaran}}\tf{
\ref{lem:chandrasekaran}, we can control the error, which arises from the directions contained in the cone $\mathcal{K}_2$.
In order to control the error, which can arise from directions in $\mathcal{K}_1$,
%(\tf{Exact same choice of words as in the last sentence on purpose?}),
we use the observation that for those directions, the nuclear norm ball around $h_0 m_0^*$ behaves locally like a euclidean ball. In particular, if the noise level $\tau$ is small, only a short segment in this direction will have smaller nuclear norm than $ h_0 m_0^* $. Hence, only a small error can occur from these near-orthogonal directions $Z \in \mathcal{K}_1 $.
\end{sproof}

\subsection{Application 6: phase retrieval with incoherence}\label{sec:phase_incoherence2}

The strategy of splitting the descent cone into two parts can also be applied to the phase retrieval problem with measurements consisting  of  arbitrary i.i.d.~subgaussian entries as introduced in Section \ref{sec_PR_incoherence}, allowing to generalize the results of \cite{krahmer2018phase} to complex-valued measurements.
%\deleted{However, in }
This is a key step towards understanding real-world applications such as ptychography, which typically do not give rise to real-valued measurements. %\deleted{ are often complex-valued.}.
For complex-valued measurements of real-valued signals,  one obtains recovery guarantees exactly analogous to those discussed in Section \ref{sec_PR_incoherence}, where this time one requires the incoherence constraint \eqref{phase:incoherence} with parameter $\mu=\tfrac{1}{81}$, see \cite[Theorem 2]{conic_subgaussian}.

The proof of these guarantees proceeds via a descent cone analysis of the cone of all admissible directions
\begin{align*}
	\mathcal{M}_{\mu} := \text{cone} \left\{ Z \in \mathcal{S}^n  : \exists x_0 \in \mathcal{X}_{\mu}\text{ such that } x_0x^*_0 +  Z\in \mathcal{S}^n_{+}   \right\},
\end{align*}
where
\begin{align*}
\mathcal{X}_{\mu}   := \left\{ x_0 \in \R^n \setminus \left\{ 0\right\} :  \Vert x_0 \Vert_{\infty} \le \mu  \Vert x_0 \Vert_2    \right\}.
\end{align*}
In order to observe how incoherence is useful, it is instructive to consider the signal $x_0= e_1=  (1, 0, \ldots, 0) \in \R^n$. Note that the matrix $Z=e_2 e_2^T -e_1 e_1^T $ is an admissible direction, that is, $ x_0x^*_0 +  t Z\in \mathcal{S}^n_{+}  $ for a sufficiently small $t>0$. However, if the measurement-vector $a_k$ satisfies 
\begin{align*}
\vert \left( a_k \right)_1 \vert = \vert \left( a_k \right)_2 \vert = \ldots = \vert \left( a_k \right)_n \vert. 
\end{align*}
we have $\trace (a_k^* Z a_k) =0 $. The problem here is that all the mass of $Z$ is concentrated on  its diagonal. The proof in \cite{conic_subgaussian} shows that this cannot be the case, if $x_0$ is incoherent.
%\tf{Ending seems a bit abrupt}

To extend the recovery guarantees 
%}\deleted{can also be extended} to complex-valued signals, \textcolor{cyan}{
one needs to address an additional difficulty. Namely, the phases of the entries of the measurement vector must be well distributed on the unit circle in $\mathbb{C}$. To see this, consider real measurements of a complex signal $x$. Then $x$ and %\deleted{$\bar x$} \ds{
$\bar{x}$
give rise to the same phaseless measurements and hence cannot be distinguished. Such ambiguities can be addressed by an additional constraints on the measurements; then the proof techniques sketched above carries over. We refer the interested reader to \cite{conic_subgaussian} for details.

\section{Conclusion}
Although many inverse problems admit a reformulation as a low-rank matrix recovery problem, as we have seen, even for the benchmark reconstruction approach via nuclear norm minimization, the structure imposed by the applications can make a significant difference.
This is true both in terms of how to analyze the reconstruction performance, as well as in terms of  the robustness results that can be expected. A key concept in this context is the role of incoherence %in the respective problem 
that distinguishes problems with comparable performance for different signals from problems where for some signals the solution is not even unique. The golfing scheme has proven to be a useful tool to derive signal-dependent recovery guarantees for incoherent signals, but has several shortcomings such as limited geometric interpretations.
Some of these shortcomings can be addressed by a refined descent cone analysis that partitions the descent cone into multiple parts that can be analyzed separately. 
To date, however, this approach has only been applied to very few scenarios, in all of which the underlying signal is of rank one. 
Generalizing this analysis to higher rank and also precisely analyzing the performance in the small noise regime would be of great importance for generating a more comprehensive understanding of the potential and limitation of low-rank matrix recovery via nuclear norm minimization. 

\section{Acknowlegments}
This work was prepared as part of the Priority Programme \emph{Compressed Sensing in Information Processing} (SPP 1798) of the German Research Foundation (DFG).

\clearpage

\begin{appendices}

\section{Appendix: Descent cone elements are effectively low rank}
 \label{sec:effective-rank}

\textbf{Lemma 2.2} \ \textit{Suppose that $Z \in \mathbb{C}^{n_1 \times n_2}$  is contained in the nuclear norm descent cone of a rank-$r$ matrix $X \in \mathbb{C}^{n_1 \times n_2}$. Then,}
\begin{align*}
	\|Z\|_* \leq \left(1+\sqrt{2}\right) \sqrt{r} \|Z\|_F.
\end{align*}

The constant $1+\sqrt{2}$ is not optimal and could be further improved by a more refined analysis. 
The argument presented here is novel and inspired by dual certificate arguments reviewed in Section~\ref{sec:golfing}. It also requires a rectangular generalization of the pinching inequality for Hermitian matrices, see e.g.\ \cite[Problem~II.5.4]{bhatia}

\addtocounter{theorem}{6}
\begin{theorem}[(Hermitian) pinching inequality] \label{thm:hermitian-pinching}
Let $P_1,\ldots,P_L \subset \mathbb{H}_n$ be a resolution of the identity ($P_l^2=P_l$ and $\sum_l P_l = \Id$). Then,
\begin{align*}
\| X \|_* \geq \sum_{l=1}^L \left\| P_l X P_l \right\|_* \quad \text{for every $X \in \mathbb{H}_n$.}
\end{align*}
\end{theorem}

We can extend pinching to general rectangluar matrices by embedding them within a larger block matrix. The \emph{self-adjoint dilation} of $Z \in \mathbb{C}^{n_1 \times n_2}$ is
\begin{align*}
\mathcal{T}(Z) = \left(
\begin{array}{cc}
0 & Z \\
Z^* & 0 
\end{array}
\right) \in \mathbb{H}_{n_1 + n_2}.
\end{align*}
Dilations preserve spectral information. In particular,
\begin{align}
\| \mathcal{T}(Z) \|_* 
%=  \mathrm{tr} \left( \sqrt{\tilde{Z}^2} \right)
=&\mathrm{tr} \left( \sqrt{\mathcal{T}(Z)^* \mathcal{T}(Z)} \right)
= \mathrm{tr} \left( 
\begin{array}{cc}
\sqrt{ZZ^*} & 0 \nonumber \\
0 & \sqrt{Z^* Z} 
\end{array}
\right) \\
=& 
\mathrm{tr} (\sqrt{Z Z^*}) + \mathrm{tr}(\sqrt{Z^* Z}) = 2 \| Z \|_*. \label{eq:dilation-nuclear-norm}
\end{align}
For simplicity, we only formulate and prove our generalization of the Hermitian pinching inequality for 
identity resolutions with two elements each. Statement and proof do, however, readily extend to more general resolutions with compatible dimensions.

\addtocounter{corollary}{1}
\begin{corollary}[Pinching for non-symmetric matrices] \label{cor:pinching}
Let $P,P^\perp \in \mathbb{H}_{n_1}$ and $Q,Q^\perp \in \mathbb{H}_{n_2}$ be two resolutions of the identity. 
Then,
\begin{align*}
\|X \|_* \geq \|P X Q \|_* + \left\| P^\perp X Q^\perp \right\|_* \quad \text{for all $X \in \mathbb{C}^{n_1 \times n_2}$.}
\end{align*}
\end{corollary}

\begin{proof}[Corollary~\ref{cor:pinching}]
Use Eq.~\eqref{eq:dilation-nuclear-norm} to relate the nuclear norm of $X$ to the nuclear norm of its self-adjoint dilation:
\begin{align*}
2 \|X \|_* = \| \mathcal{T}(X) \|_*
= \left\| 
\left(
\begin{array}{cc}
0 & X \\
X^* & 0
\end{array}
\right)
\right\|_*.
\end{align*}
Next, we combine $P,P^\perp \in \mathbb{H}_{n_1}$ and $Q,Q^\perp \in \mathbb{H}_{n_2}$ to obtain a resolution of the identity with compatible dimension:
\begin{align*}
\left(
\begin{array}{cc}
P & 0 \\
0 & Q
\end{array}
\right), \;
\left(
\begin{array}{cc}
P^\perp & 0 \\
0 & Q^\perp
\end{array}
\right) \in \mathbb{H}_{n_1+n_2}
\end{align*}
Since everything is Hermitian,
we can apply Theorem~\ref{thm:hermitian-pinching} (original pinching) with respect to this resolution of the identity to the nuclear norm of the s.a.~dilation:
\begin{align*}
\left\|
\left(
\begin{array}{cc}
0 & X \\
X^* & 0
\end{array}
\right)
\right\|_*
\geq & 
\left\|
\left(
\begin{array}{cc}
P & 0 \\
0 & Q 
\end{array}
\right)
\left(
\begin{array}{cc}
0 & X \\
X^* & 0
\end{array}
\right)
\left(
\begin{array}{cc}
P & 0 \\
0 & Q 
\end{array}
\right)
\right\|_*
+ 
\left\|
\left(
\begin{array}{cc}
P^\perp & 0 \\
0 & Q^\perp 
\end{array}
\right)
\left(
\begin{array}{cc}
0 & X \\
X^* & 0
\end{array}
\right)
\left(
\begin{array}{cc}
P^\perp & 0 \\
0 & Q^\perp 
\end{array}
\right)
\right\|_*\\
= & \left\|
\left(
\begin{array}{cc}
0 & PXQ \\
QX^*P & 0
\end{array}
\right)
\right\|_* + \left\|
\left(
\begin{array}{cc}
0 & P^{\perp}XQ^{\perp} \\
Q^{\perp}X^*P^{\perp} & 0
\end{array}
\right)
\right\|_*.
\end{align*}\\
We can now recognize self-adjoint dilations of two rectangular matrices. Using Eq.~\eqref{eq:dilation-nuclear-norm} implies
\begin{align*}
\| \mathcal{T}(X) \|_*
\geq & \| \mathcal{T}(PXQ) \|_* + \| \mathcal{T}(P^\perp X Q^\perp) \|_*
= 2 \| PXQ \|_* + 2 \|P^\perp X Q^\perp \|_*
\end{align*}
\
\end{proof}

Next, the concept of sign functions of real numbers is extendable
to non-Hermitian matrices. Let $X \in \mathbb{C}^{n_1 \times n_2}$ be a rectangular matrix with SVD $X = U \Sigma V^*$. 
We define its sign matrix to be $\mathrm{sign}(X) = U V^* \in \mathbb{C}^{n_1 \times n_2}$. Note that this sign matrix is unitary and obeys
\begin{align*}
\langle \mathrm{sign}(X),X \rangle_F = \mathrm{tr} \left( (UV^*)^* U \Sigma V^* \right) = \mathrm{tr}(\Sigma) = \|X \|_*.
\end{align*}
The last ingredient is the dual formulation of the nuclear norm:
\begin{align*}
\| X \|_* = \max_{\|U\| \leq 1}\left| \langle U,X \rangle \right| = \max_{U \mathrm{ unitary}} \left| \langle U, X \rangle \right|.
\end{align*}

\begin{proof}[Lemma~\ref{lem:effective-rank}]
By assumption, $Z \in \mathbb{C}^{n_1 \times n_2}$ is contained in the descent cone of a rank-$r$ matrix $X$. This implies that there exists $\tau>0$ such that $\|X \|_* \geq \| X + \tau Z \|_*$. 
Apply a SVD $X=U\Sigma V^*$ and use it to define $r$-dimensional orthoprojectors $P=UU^* \in \mathbb{H}_{n_1}$, $Q=VV^* \in \mathbb{H}_{n_2}$, as well as their orthocomplements $P^\perp =\Id-P$ and $Q^\perp = \Id-Q$.
Use them to define the matrix-valued projections
\begin{align*}
\mathcal{P}_{T_X}^\perp:\; Z \mapsto P^\perp Z Q^\perp \quad \text{and} \quad \mathcal{P}_{T_X}:\; \mapsto Z-\mathcal{P}_{T_X}^\perp (Z) = PZ + ZQ - PZQ
\end{align*}
such that $Z=\mathcal{P}_{T_X}^\perp (Z) + \mathcal{P}_{T_X} (Z) = Z_{T_X}^\perp + Z_{T_X}$ and, in particular, $X_{T_X}^\perp =0$ and $X_{T_X} = X$.
In words, $\mathcal{P}_{T_X}$ projects $\mathbb{C}^{n_1 \times n_2}$ onto a subspace whose compression to the kernel of $X$ vanishes identically, namely the tangent space of $X$ (as defined in (\ref{equ:deftangentspace})).
Moreover, for every $Z \in \mathbb{C}^{n_1 \times n_2}$
\begin{align}
\mathrm{rk} \left( Z_{T_X} \right)=& \mathrm{rk} \left( PZ+(P+P^\perp) ZQ-PZQ \right)
= \mathrm{rk} \left( PZ + P^\perp ZQ \right) \nonumber \\
\leq & \mathrm{rk} \left( PZ \right) + \mathrm{rk} \left( P^\perp Z Q \right) 
\leq \mathrm{rk}(P) + \mathrm{rk}(Q) =2r,
\label{eq:projection-rank}
\end{align}
because matrix rank is subadditive and cannot increase under matrix products, for every $Z \in \mathbb{C}^{n_1 \times n_2}$.
Corollary~\ref{cor:pinching} (pinching) -- with respect to $P$ and $Q$ --  and the descent cone property of $Z$ together imply 
\begin{align*}
\| X \|_* \geq & \left\| X + \tau Z \right\|_* \geq \left\| P (X+\tau Z) Q \right\|_* + \left\| P^\perp (X + \tau Z) Q^\perp \right\|_* \\
=& \left\| X + \tau P Z Q \right\|_* + \tau \left\| P^\perp Z Q^\perp \right\|_* \\
=& \left| \langle \mathrm{sign}(X+\tau PZQ), X + \tau PZQ \rangle_F \right| + \tau \left\| P^\perp Z Q^\perp \right\|_* \\
\geq & \left| \langle \mathrm{sign} (X), X \rangle_F + \tau \langle \mathrm{sign}(X),P ZQ \rangle_F \right| + \tau \|P^\perp Z Q^\perp \|_* \\
\geq &\| X \|_* + \tau \left(- \left| \langle \mathrm{sign}(X),PZQ\rangle_F \right|+ \left\| P^\perp Z Q^\perp \right\|_* \right).
\end{align*}
Since $\tau >0$, this chain of inequalities can only be valid if 
\begin{align*}
\left\| Z_{T_X}^\perp \right\|_* = \left\| P^\perp Z Q^\perp \right\|_* \leq \left| \langle \mathrm{sign}(X),P Z Q \rangle_F \right| \leq \| \mathrm{sign}(X) \| \| PZQ \|_* \leq \sqrt{r} \| PZQ \|_F,
\end{align*}
because both $P$ and $Q$ are rank-$r$ projectors. We can combine this with a decomposition $Z=Z_{T_X}^\perp + Z_{T_X}$ and Rel.~\eqref{eq:projection-rank} to conclude
\begin{align*}
\| Z \|_* \leq & \left\| Z_{T_X}^\perp \right\|_* + \left\| Z_{T_X} \right\|_*
\leq \sqrt{r} \|P Z Q \|_F + \sqrt{\mathrm{rank}(Z_{T_X})} \|Z_{T_X} \|_F \\
\leq & \sqrt{r} \|Z \|_F + \sqrt{2r} \|Z \|_F = \left(1+\sqrt{2}\right) \sqrt{r} \|Z \|_F,
\end{align*}
because both $Z \mapsto PZQ$ and $Z \mapsto Z_{T_X}$ are contractions with respect to the Frobenius norm. 
\end{proof}
\end{appendices}
	
\clearpage
\bibliographystyle{amsalpha}
\bibliography{main-arxiv}
	
\end{document}